\pdfoutput=1

\documentclass[aps,reprint]{revtex4-1}

\usepackage{graphicx}
\usepackage{dcolumn}
\usepackage{natbib,aas_macros}
\usepackage{color}
\usepackage{amssymb}
\usepackage{amsthm}
\usepackage{mathrsfs}
\usepackage{amsmath}
\usepackage{hyperref}
\usepackage{lineno}
\usepackage{tikz}
\usepackage[normalem]{ulem}

\usepackage{comment}

\usepackage{soul}

\usetikzlibrary{external}
\usetikzlibrary{positioning}
\newcommand{\mb}{\mathbf}
\newcommand{\bb}{\mathbb}
\newcommand{\innerprod}[2]{\left\langle #1, #2 \right\rangle}

\newcommand{\R}{\bb R}
\newcommand{\E}{\bb E}
\newcommand{\prob}{\bb P}
\newcommand{\ST}{\mathrm{subject\ to}}
\renewcommand{\mathbf}{\boldsymbol}
\newcommand{\mr}{\mathrm}
\newcommand{\mc}{\mathcal}

\newcommand{\fnr}{\mathrm{FNR}}
\newcommand{\fpr}{\mathrm{FPR}}
\newcommand{\wh}{\widehat}

\newcommand{\rhonoise}{\rho_{\mr{noise}}}

\newtheorem{theorem}{Theorem}
\newtheorem{definition}{Definition}

\newtheorem{proposition}[theorem]{Proposition}

\newcommand{\shallownet}{{\tt MNet-Shallow}}
\newcommand{\deepnet}{{\tt MNet-Deep}}

\newcommand{\wfnr}{\mathrm{WFNR}}

\graphicspath{{./figs/}}

\begin{document} 
\title{Generalized Approach to Matched Filtering using Neural Networks}


\author{Jingkai Yan$^{1,4}$, Mariam Avagyan$^{1,4}$, Robert E. Colgan$^{2,4}$, Do\u{g}a Veske$^{3}$, Imre Bartos$^{6}$, John Wright$^{1,4}$, Zsuzsa M\'arka$^{4,5}$, and Szabolcs M\'arka$^{3,4}$\vspace{.1in}}

\affiliation{$^1$Department of Electrical Engineering, Columbia University in the City of New York, 500 W. 120th St., New York, NY 10027, USA\\
$^2$Department of Computer Science, Columbia University in the City of New York, 500 W. 120th St., New York, NY 10027, USA\\
$^3$Department of Physics, Columbia University in the City of New York, 538 W. 120th St., New York, NY 10027, USA\\
$^4$Data Science Institute, Columbia University in the City of New York, 550 W. 120th St., New York, NY 10027, USA\\
$^5$Columbia Astrophysics Laboratory, Columbia University in the City of New York, 538 W. 120th St., New York, NY 10027, USA\\
$^6${Department of Physics, University of Florida, PO Box 118440, Gainesville, FL 32611-8440, USA}
}

\begin{abstract}

Gravitational wave science is a pioneering field with rapidly evolving data analysis methodology currently assimilating and inventing deep learning techniques. The bulk of the sophisticated flagship searches of the field rely on the time-tested matched filtering principle within their core. In this paper, we make a key observation on the relationship between the emerging deep learning and the traditional techniques: \emph{matched filtering is formally equivalent to a particular neural network}. This means that a neural network can be constructed analytically to exactly implement matched filtering, and can be further trained on data or boosted with additional complexity for improved performance. Moreover, we show that the proposed neural network architecture can outperform matched filtering, both with or without knowledge of a prior on the parameter distribution. When a prior is given, the proposed neural network can approach the statistically optimal performance. We also propose and investigate two different neural network architectures \shallownet\ and \deepnet, both of which implement matched filtering at initialization and can be trained on data. \shallownet\ has simpler structure, while \deepnet\ is more flexible and can deal with a wider range of distributions. Our theoretical findings are corroborated by experiments using real LIGO data and synthetic injections, where our proposed methods significantly outperform matched filtering at false positive rates above $5\times 10^{-3}\%$. The fundamental equivalence between matched filtering and neural networks allows us to define a ``complexity standard candle'' to characterize the relative complexity of the different approaches to gravitational wave signal searches in a common framework. Additionally it also provides a glimpse of an intriguing symmetry that could provide clues on interpretability, namely how neural networks approach the problem of finding signals in overwhelming noise. Finally, our results suggest new perspectives on the role of deep learning in gravitational wave detection.

\end{abstract}
\maketitle

\section{Introduction}
\label{sec:intro}

The discovery of cosmic gravitational waves~\cite{PhysRevLett.116.061102}, the windfall of binary black-hole (BBH) merger detections~\cite{Abbott_2019,2020arXiv201014527A}, and the spectacular insights that multimessenger astrophysics provided~\cite{PhysRevLett.119.161101,2017ApJ...848L..12A} revolutionized how we understand the Universe. This leap was due to multiple factors, from instrumental advances to computing breakthroughs. Emerging interferometric gravitational wave detectors, KAGRA~\cite{10.1093/ptep/ptaa125}, GEO600~\cite{Dooley_2016}, Virgo~\cite{2015CQGra..32b4001A}, and LIGO~\cite{1992Sci...256..325A,2015CQGra..32g4001L}, played a critical role as they provided the technology~\cite{Affeldt_2014,PhysRevLett.123.231108,PhysRevLett.123.231107} enabling signals to be extracted from ripples in Einstein's space-time~\cite{1916SPAW.......688E,1918SPAW.......154E}. Of course, as it is not sufficient to have data with faint cosmic signals buried in the noise, the community had to rely on exquisitely sensitive data analysis algorithms to extract transient signals from the noisy data. The bulk of the discoveries were made by two classes of powerful data analysis approaches, 
\emph{excess power}~\cite{2001PhRvD..63d2003A,2004CQGra..21S1819K,2000IJMPD...9..303A} and 
\emph{matched filtering}~\cite{1989thyg.book.....H,PhysRevD.60.022002,1993PhRvL..70.2984C,1994PhRvD..49.2658C,1998PhRvD..57.4535F,1998PhRvD..57.4566F,2004PhRvD..69l2001A}.
The flagship matched filtering methods~\cite{PyCBCSoft,2012PhRvD..85l2006A,2019PASP..131b4503B,2005PhRvD..71f2001A,2014PhRvD..90h2004D,2016CQGra..33u5004U,2018PhRvD..98b4050N,2017PhRvD..95d2001M,2019arXiv190108580S,2020PhRvD.101b2003H,76e97cf9801544919973534ed7028b6a,2016CQGra..33q5012A, nitz2017detecting} reached unprecedented sophistication and became the workhorse of the field~\cite{Abbott_2019,2020arXiv201014527A}. 
Insightful work also exist on the extent of optimality, role of intrinsic parameters, and effect of non-Gaussian backgrounds ~\cite{2008arXiv0804.1161S,2012PhRvD..85l2008B,2014PhRvD..89f2002D}.
There is more than historical evidence on their algorithmic power~\cite{1999PhRvD..60b2002O}, and they are also considered optimal~\cite{1994PhRvD..49.2658C} when searching for chirps of known shape~\cite{1996PhRvD..53.6749O,PhysRevD.60.022002, 1991PhRvD..44.3819S,1994PhRvD..49.1707D} embedded in well-behaved Gaussian noise. Within the optimality and success lie limitations, as the data is significantly more complex~\cite{2016CQGra..33m4001A,2020CQGra..37e5002A} than Gaussian noise and many cosmic signals are not as well known as the BBH models that are being used in searches~\cite{2020arXiv200905461G}. Therefore, it is critical that we both seek data analysis methods beyond the horizon of current techniques and rigorously understand the place of current techniques in the broader field of possible methods.

An abundance of prior works has been using deep learning methods for gravitational wave detection. Convolutional neural networks have been shown to be capable of identifying gravitational waves and their parameters from binary black holes and binary neutron stars, with performance approaching the matched filtering search currently used by LIGO, Virgo and KARGA \cite{gebhard2017convwave, george2018deep, gabbard2018matching, george2018deepneural, fan2019applying, morawski2020convolutional, dreissigacker2019deep, krastev2020real, lin2020binary, lin2020detection, bresten2019detection, astone2018new, yamamoto2020use, dreissigacker2020deep, corizzo2020scalable, miller2019effective, bayley2020robust, krastev2020detection, luo2020extraction, santos2020gravitational, chan2020detection, xia2021improved, 2015CQGra..32b4001A, 2015CQGra..32g4001L}. 
In addition, these machine learning (ML) method can also be applied to glitches and noise transients identification \cite{biswas2013application, george2017deep, mukund2017transient, razzano2018image, fan2019applying, coughlin2019classifying, colgan2020efficient}, signal classification and parameter estimation \cite{nakano2019comparison, green2020gravitational, marulanda2020deep, caramete2020characterization, delaunoy2020lightning}, data denoising \cite{shen2019denoising, wei2020gravitational}, etc. While these works exhibit neural networks that could approach the performance of matched filtering, they are still often applied as or considered ``black box'' models. This makes it challenging to evaluate the statistical evidence provided by neural networks, and to incorporate that evidence in downstream analyses \cite{gebhard2019convolutional}.

This paper is motivated by a critical observation, which we substantiate below: {\em matched filtering with a collection of templates is formally equivalent to a particular neural network}, whose architecture and parameters are dictated by the templates.
This observation has precedents in the machine learning literature, where deep neural networks are sometimes viewed as hierarchical template matching methods, with signal-dependent, class-specific templates \cite{balestriero2018mad,cheng2019qatm,bower1988neural,1085953,buniatyan2017deep,xue1992neural,lippmann1989adaptive}. Here, we delineate a simple and explicit equivalence between matched filtering and particular neural networks, which can be constructed analytically from a set of templates. This equivalence lies in the algorithmic level, and does not depend on specific problem formulations. 

In order to study the potential performance gains of using neural networks, we formulate the gravitational wave detection problem abstractly as the detection of a parametric family of signals. Under this framework, we show that the analytically constructed networks can also be used as a principled starting point for learning from data, yielding signal classifiers with better performance than their initialization, namely ``standing on the shoulder of giants''. Such learning can be applied to scenarios both with or without a prior distribution on the parameters. In particular, when a prior distribution is given, we show that the learned neural network can (empirically) approach the statistically optimal performance. 

We propose and investigate two different neural network architectures for implementing matched filtering, respectively {\shallownet} and {\deepnet}. The former has simpler structure, while the latter is more flexible and can deal with a wider range of distributions. These learned classifiers have a number of additional advantages: they do not require prior knowledge of the noise distribution, can be adapted to cope with time-varying noise distributions, and suggest new approaches to computationally efficient signal detection. 
We conducted experiments using real LIGO data~\cite{ABBOTT2021100658} in order to demonstrate the feasibility and power of neural networks in comparison to matched filtering, where we validate our findings empirically  that neural networks via training can reach better performance.
Finally, interpreting matched filtering and neural networks in a common framework also allows a clear comparison of their computational/storage complexities and statistical strengths, consequently making deep-learning less of a mystery. 

The rest of the paper is organized as follows. Section II introduces the problem of parametric signal detection as an abstraction of the gravitational-wave detection problem, and discusses the two formulations of the objective. Section III discusses matched filtering as an approach to solving the parametric detection problem, as well as its limitations. Section IV illustrates how neural network models can be applied in this problem, in a way that exactly implements matched filtering at initialization. Section V discusses the training process of neural network models, and in particular how it is aligned with the parametric signal detection problem. In Section VI we present experimental results on real LIGO data and synthetic injections. We discuss some further implications of this work in Section VII, and conclude in Section VIII.

\section{Parametric Signal Detection}
\label{sec:problem}

The problem of identifying gravitational waves~\cite{PhysRevD.46.5236} in a single gravitational-wave detector data stream~\footnote{In general, a global Earth and Space based gravitational-wave detector network can be treated as a composite data stream~\citep{1999PhRvD..60b2002O, 2001PhRvD..63j2001F}. However, that added complexity is not necessary when discussing the principles of the paper. Therefore, we constrain ourselves to a single datastream in this proof of principle analysis.} can be formulated as follows: we observe detector strain data $\mb x \in \bb R^n$, and wish to determine whether $\mb x$ consists of astrophysical signal plus noise, or noise alone. We can model possible astrophysical signals as belonging to a parametric family
\begin{equation}
    S_\Gamma = \{ \mb s_{\mb \gamma} \mid \mb \gamma \in \Gamma \}
\end{equation}
where the parameters $\mb \gamma$ can represent properties of the objects that generate the gravitational wave, such as masses, orbits and spins. We assume the signals are normalized to have unit 
power, namely $\|\mb s_{\mb \gamma}\|^2=1$ for all $\mb \gamma$. We model noise as a random vector $\mb z \in \bb R^{n}$, which is assumed to follow distribution $\rho_0$ and be probabilistically independent of the signal. In this notation, our goal becomes one of solving a hypothesis testing problem:
\begin{eqnarray}
    &H_0:& \mb x = \mb z, \\
    \text{or} \quad &H_1:& \mb x = \mb s_{\mb \gamma} + \mb z \ \text{ for some } \ \mb  \gamma \in \Gamma.
\end{eqnarray}
Note that except for special cases, such as when the hypothesis $H_1$ is simple, or when the parameters associated with $H_1$ satisfy certain monotone conditions, we usually do not have a uniformly most powerful test \cite{casella2021statistical}.

Our broad goal is to identify decision rules $\delta: \mathbb{R}^n \to \{ 0, 1 \}$ that (i) have good statistical performance and (ii) can be implemented efficiently. Our approach will start with analytically defined neural networks, which precisely replicate matched filtering, and then train these networks to optimize their statistical performance. We will give training approaches
that are compatible with two classical frameworks for formalizing the performance decision rules $\delta$: the {\em Neyman-Pearson} framework, in which the parameter $\mb \gamma$ is a random vector with known distribution $\nu$, and the {\em minimax} framework, in which we control the worst performance over all possible choices of the parameter $\mb \gamma$. 

\subsection{Neyman-Pearson Framework} 

In this setting, one assumes that $\mb \gamma$ is a random vector with probability distribution $\nu$. With this distribution $\nu$, we can then view $H_1$ as a simple hypothesis. The false positive rate (FPR) associated with the rule $\delta$ is 
\begin{align}
    \fpr = \bb P_{\mb z}\left[ \delta(\mb z) =1  \right] 
\end{align}
The false negative rate (FNR) at signal $\mb s_{\mb \gamma}$ is 
\begin{equation}
    \fnr_{\mb \gamma} = \bb P_{\mb z} \left[ \delta\left( \mb s_{\mb \gamma}+ \mb z \right) = 0 \right].
\end{equation}
The {\em overall} false negative rate is 
\begin{equation}
    \fnr = \int \fnr_{\mb \gamma} \, \mr{d} \nu(\mb \gamma). 
\end{equation}
The Neyman-Pearson criterion seeks the optimal tradeoff between $\fnr$ and $\fpr$: 
\begin{align}
    \min_{\delta} \; \fnr \quad \ST \quad \fpr \le \alpha,
\end{align}
where $\alpha$ is a user-specified significance level. 

There is a classical closed form expression for the optimal test under the Neyman-Pearson criterion: if $\rho_0$ and $\rho_1$ are the probability densities of the signal $\mb x$ under hypotheses $H_0$ and $H_1$, respectively, then the optimal test is given by comparing the {\em likelihood ratio}
\begin{equation}
    \lambda(\mb x) = \frac{\rho_1(\mb x)}{\rho_0(\mb x)}
\end{equation}
to a threshold $\tau$, which depends on the significance level $\alpha$. An illustration of an example problem is shown in FIG \ref{fig:setup}.

\subsection{Minimax Framework} 

When a good prior $\nu$ is not available or cannot be assumed, we can instead seek a decision rule that solves 
\begin{equation} \label{eqn:minimax} 
    \min_{} \; \wfnr \quad \ST \quad \fpr \le \alpha. 
\end{equation}
at a given false positive rate, where $\wfnr$ is the {\em worst false negative rate}  defined as
\begin{equation} 
    \wfnr = \max_{\mb \gamma \in \Gamma} \fnr_{\mb \gamma}.
\end{equation}
In contrast to the Neyman-Pearson criterion, there is in general no simple expression for the minimax optimal rule $\delta$ \cite{yu1992general}. In the next section, we will review matched filtering, a simple, popular approach to detection which is compatible with the minimax framework (albeit suboptimal in terms of \eqref{eqn:minimax}), in the sense that it does not require a prior on $\mb \gamma$.

\begin{figure}
    \centering
    \includegraphics[width=3in, trim={0 .2in 0 .4in}]{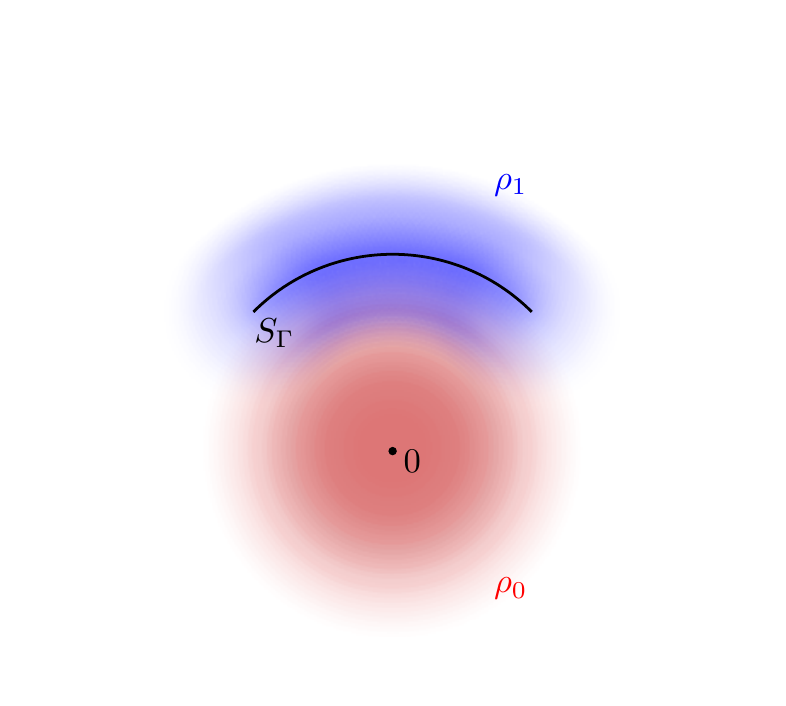}
    \caption{An example of the parametric signal detection problem with signal space $S_\Gamma$. Densities $\rho_0$ and $\rho_1$ are shown in red and blue respectively.}
    \label{fig:setup}
\end{figure}

\section{Matched Filtering for Parametric Detection}
\label{sec:mf}

{\em Matched filtering} is a powerful classical approach to signal detection, which applies a linear filter which is chosen to maximize the signal-to-noise ratio (SNR). 

\subsection{Optimality for Single Signal Detection}

In the simplest possible setting, in which (i) there is only one target signal $\mb s$, (ii) the observation $\mb x$ has the same length as $\mb s$, and (iii) the noise is uncorrelated (i.e., $\bb E\left[ \mb z \mb z^* \right] = \sigma^2\mb I$), matched filtering simply computes the inner product between the target $\mb s$ and the observation:
\begin{equation}
    \delta(\mb x)=1 \text{ iff } \left<\mb s, \mb x\right> \ge \tau.
    \label{eqn:mf-single}
\end{equation}
When detecting a single signal $\mb s$ in iid Gaussian noise, this decision rule is optimal in both the Neyman-Pearson and minimax senses: for example, if $\mb z \sim \mc N(\mb 0,\sigma^2\mb I)$, 
the likelihood ratio 
\begin{equation}
    \lambda(\mb x) = \frac{\rho_0(\mb x - \mb s)}{\rho_0(\mb x)} = \exp \left(\frac{\left<\mb s,\mb x\right> - \|\mb s\|^2/2}{\sigma^2} \right)
\end{equation} 
is a monotone function of $\innerprod{\mb s}{\mb x}$, and so matched filtering implements the (optimal) likelihood ratio test. FIG \ref{fig:mf-opt} illustrates this optimality geometrically. 

The simplicity and optimality in this setting make matched filtering a principled choice for signal detection, and have inspired its application in settings that go far beyond the scope of this rigorous guarantee. In particular, the simplest and most practical extension of this rule to detecting parametric families of signals $\mb s_{\mb \gamma}$ is suboptimal in both the Neyman-Pearson and minimax settings. Moreover, there are a number of additional factors which contribute to its suboptimality. These include unknown, non-Gaussian, and possibly time-varying noise distributions as well as density and coverage issues in the template bank, which for complexity reasons may cover only a small portion of the phase space \cite{1994PhRvD..49.2658C}. Nevertheless, we will see how matched filtering can inspire principled approaches to deriving more flexible decision rules which can address many of these challenges. 

\subsection{Extensions to Parametric Detection}

The simplest extension of the decision rule \eqref{eqn:mf-single} to  {\em parametric} detection problems, in which there are multiple potential targets $\mb s_{\mb \gamma}$, involves taking the maximum over the parameter space:
\begin{equation}
    \delta(\mb x)=1 \text{ iff } \max_{\mb \gamma \in \Gamma} \left<\mb s_{\mb\gamma}, \mb x\right> \ge \tau.
    \label{eqn:mf-multi}
\end{equation}
Here we used the assumption that all templates have unit norm, namely $\|\mb s_{\mb \gamma}\|_2^2=1,\ \forall \mb\gamma\in\Gamma$. When this rule \eqref{eqn:mf-multi} is hard to implement in exact form, it can typically be approximated by taking samples $\mb s_{\mb \gamma_1},\dots,\mb s_{\mb \gamma_k}$ and setting
\begin{equation}
    \delta(\mb x)=1 \text{ iff } \max_{i=1,\dots,k} \left<\mb s_{\mb \gamma_i}, \mb x\right> \ge \tau.
    \label{eqn:mf-finite}
\end{equation}
When the sampling is sufficiently dense, the sampled matched filter rule \eqref{eqn:mf-finite} accurately approximates the ideal matched filter rule \eqref{eqn:mf-multi} \cite{1994PhRvD..49.2658C}. This rule, while simple, is an important component of many sophisticated data analysis pipelines, including LIGO, Virgo and KARGA's template based searches for compact binary coalescence signals. 

Note that the matched filtering decision rule \eqref{eqn:mf-multi} has connections to the (generalized) likelihood ratio test, where $H_1$ is the composite hypothesis  $\mb s_{\mb \gamma}\in S_\Gamma$. While this test has nice statistical properties, it is not guaranteed to be the uniformly most powerful test when the hypotheses are composite. For the rest of this paper, the term ``likelihood ratio test'' will be reserved for the test with a given prior and simple hypotheses, which satisfies the Neyman-Pearson criterion.

In contrast to the single signal setting, the simple extensions \eqref{eqn:mf-multi}-\eqref{eqn:mf-finite} of matched filtering to detecting parametric families of signals are not optimal: in the Neyman-Pearson setting, they do not achieve the minimal $\fnr$ for a given $\fpr$, while in the minimax setting, they do not achieve the minimal $\wfnr$ for a given $\fpr$.

The suboptimality of \eqref{eqn:mf-multi}-\eqref{eqn:mf-finite} under Neyman-Pearson can be observed by noting that the decision statistic $\max_{\mb \gamma} \innerprod{\mb s_{\mb \gamma}}{\mb x}$ is not a monotone function of the likelihood ratio, which in iid Gaussian noise for example, takes the form 
\begin{equation}
    \lambda(\mb x) = \int \exp\left( \frac{ \innerprod{\mb s_{\mb \gamma}}{\mb x} - \| \mb s_{\mb \gamma} \|^2 / 2 }{\sigma^2} \right) \mr{d}\nu(\mb \gamma). 
\end{equation}

FIG \ref{fig:mf-subopt} and \ref{fig:mf-subopt-roc} illustrate such suboptimality for a particular problem configuration in $\bb R^2$. Note that throughout our paper, we will slightly abuse the term of
receiver operating characteristic (ROC) curves by plotting $\fnr$ against $\fpr$, instead of the convention of plotting $\fpr$ against the true positive rate $\mr{TPR}\equiv 1-\fnr$. This highlights the connection to the notion of error rates in machine learning, and more importantly facilitates demonstration of the curves and axis ranges at very low error rates.

It is, in a sense, unsurprising that matched filtering is suboptimal in this setting, since the decision rules \eqref{eqn:mf-multi}-\eqref{eqn:mf-finite} do not make use of the prior $\nu$, while the likelihood ratio test assumes (and uses) this prior.

However, the matched filtering rule \eqref{eqn:mf-multi}-\eqref{eqn:mf-finite} is also in general suboptimal in the ``prior-free'' minimax setting. Consider the scenario in FIG \ref{fig:mf-subopt-minimax} as an example, where the signal space $S_\Gamma \subset\R^2$ consists of only two signals $\mb s_1=[1,0]^T$ and $\mb s_2=[0,1]^T$. Comparing the prior-free matched filtering decision rule $\delta_{\text{MF}}$ with the optimal decision rule $\delta_*$ under the Neyman-Pearson framework with uniform prior over the two signals, we see that $\delta_{\text{MF}}$ is suboptimal under Neyman-Pearson criterion with uniform prior. Moreover, from symmetry it follows that for symmetric decision rules such as $\delta_{\text{MF}}$ and $\delta_*$ the worst FNR and the overall FNR are equal. This implies that $\delta_{\text{MF}}$ is also worse than $\delta_*$ under the minimax criterion.

We also note that this suboptimality is, in some sense, not because we don't have sufficient templates. In the example shown in FIG \ref{fig:mf-subopt-minimax}, the matched filtering model already covers the entire signal set which consists of two signals. Furthermore, we will see in the later discussions that matched filtering has other structural limitations when working with non-Gaussian noise distributions.

\begin{figure}
    \centering
    \includegraphics[width=2in, trim={0 .2in 0 .2in}]{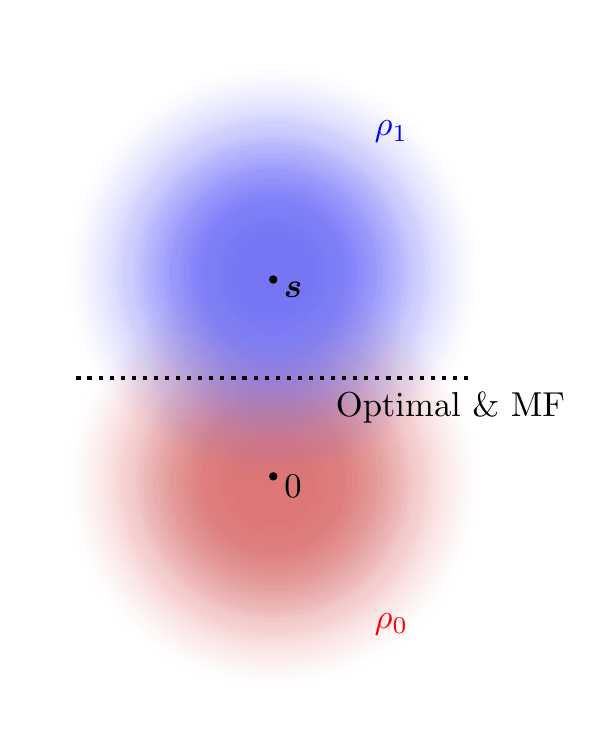}
    \caption{Optimality of matched filtering in single signal detection.}
    \label{fig:mf-opt}
\end{figure}

\begin{figure}
    \centering
    \includegraphics[width=3in, trim={0 .15in 0 .5in}]{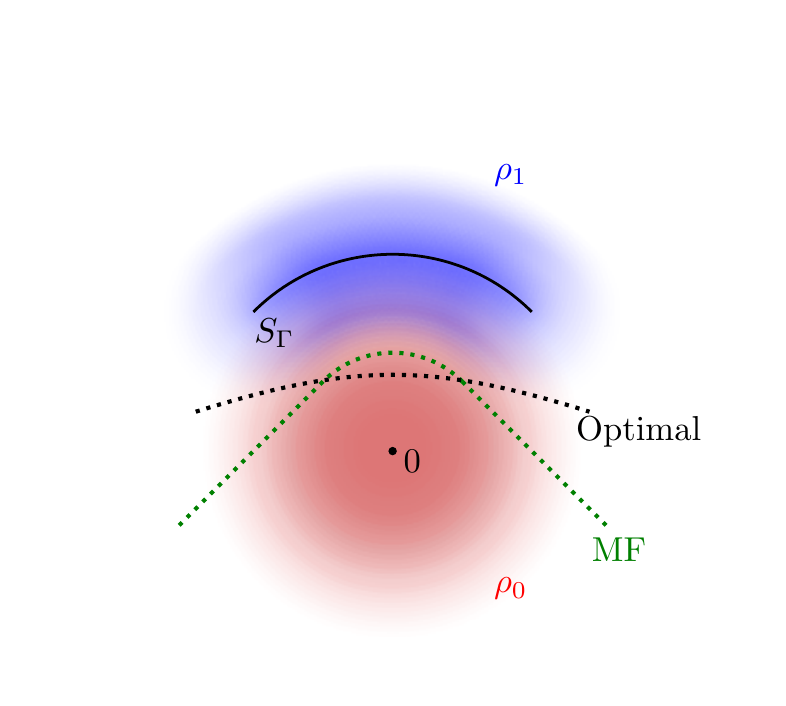}
    \caption{Suboptimality of matched filtering under the Neyman-Pearson framework.}
    \label{fig:mf-subopt}
\end{figure}

\begin{figure}
    \centering
    \includegraphics[width=3in]{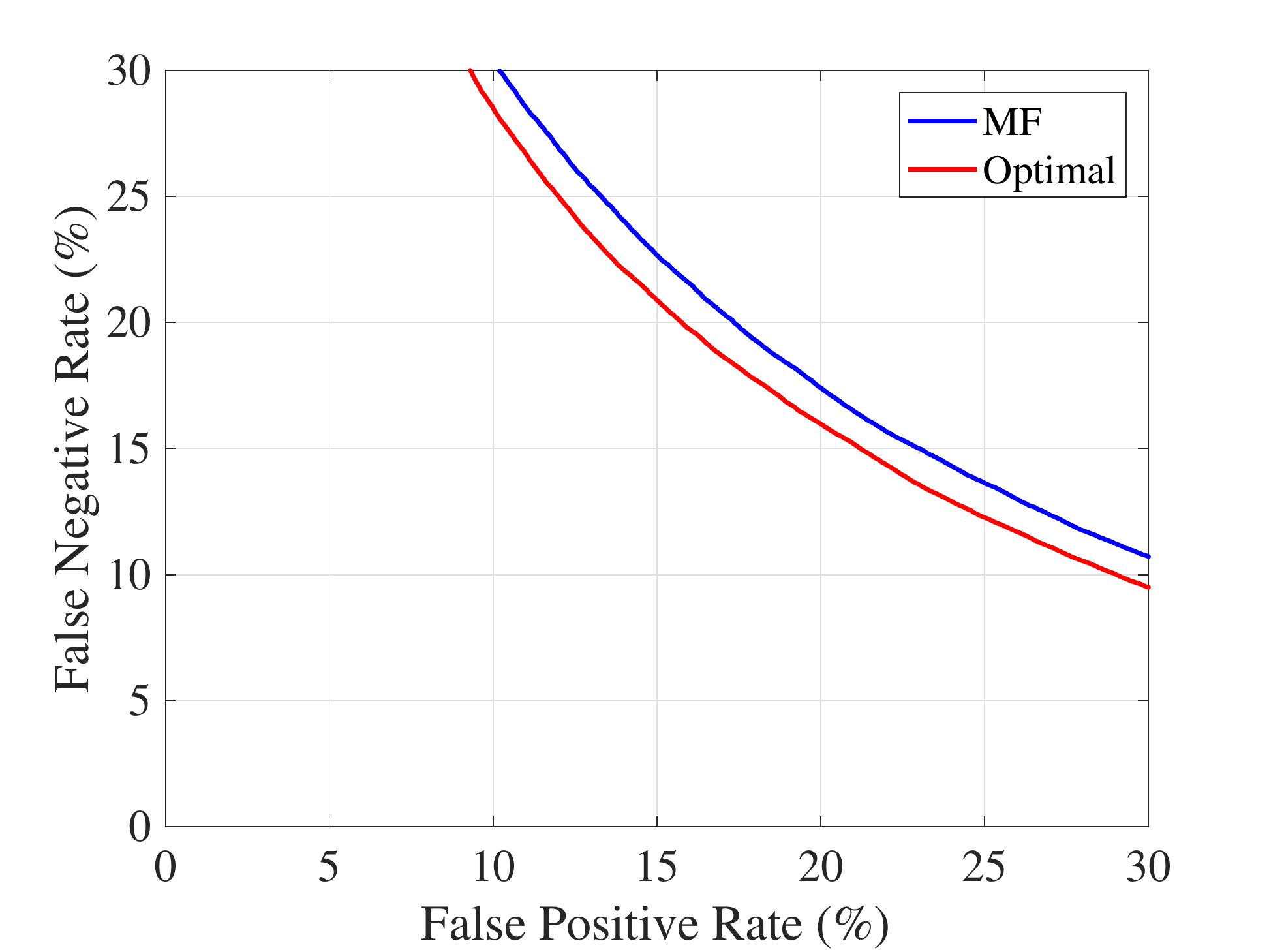}
    \caption{Comparison of ROC curves of the optimal classifier and matched filtering in the 2-dimensional concept as illustrated in FIG \ref{fig:mf-subopt}.}
    \label{fig:mf-subopt-roc} 
\end{figure}

\begin{figure}
    \centering
    \includegraphics[width=2.5in, trim={0 0 0 0}]{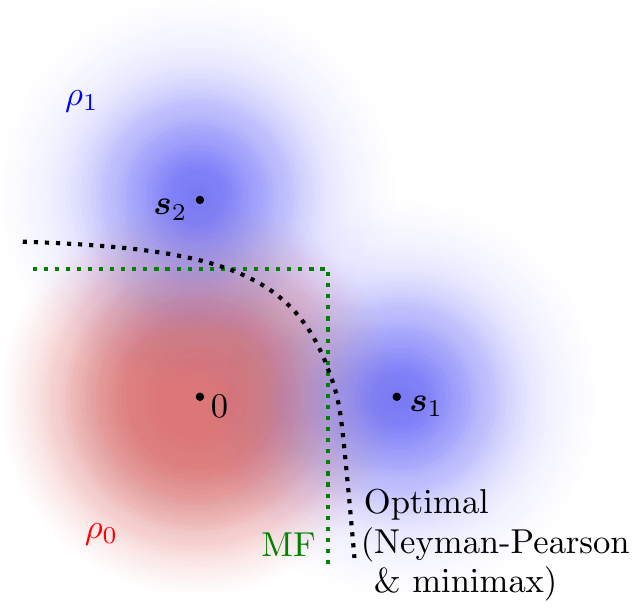}
    \caption{Suboptimality of matched filtering under the minimax framework.}
    \label{fig:mf-subopt-minimax}
\end{figure}

\section{From Matched Filtering to Neural Networks}

Since the matched filtering rule \eqref{eqn:mf-finite} is suboptimal for parametric detection, we will show that (i) the form of this rule suggests approaches to learning optimal rules for parametric detection, and (ii) the resulting classifiers have additional advantages, including greater flexibility and lower computational/storage complexity or cost. Our approach is driven by the observation: the matched filtering rule \eqref{eqn:mf-finite} is equivalent to a feedforward neural network. 

\subsection{Neural Networks: Notation and Basics}

A {\em neural network} implements a mapping from the signal space $\bb R^n$ to an output space $\bb R^d$:
\begin{equation}
    f_{\mb \theta} : \bb R^n \to \bb R^d.
\end{equation}
Here, $\mb \theta$ represents the parameters of the network. Specifically, a fully connected neural network can be written as a composition of layers, each of which applies an affine mapping \begin{equation}
    \mb x \mapsto \mb W \mb x + \mb b
\end{equation}
followed by an elementwise activation function $\phi$: 
\begin{eqnarray}
    f_{\mb \theta}(\mb x) &=& \mb W^L \phi\Bigl( \mb W^{L-1} \phi\Bigl( \dots \phi \Bigl(\mb W^1 \mb x + \mb b^1 \Bigr) \nonumber \\
    && \qquad \qquad \qquad \qquad \dots \Bigr) + \mb b^{L-1} \Bigr) + \mb b^L.
    \label{eqn:nn-general-form}
\end{eqnarray}
With slight abuse of notation, the activation function $\phi:\R\to\R$ acts elementwise when applied to a vector:
\begin{equation}
    \phi([v_1,\dots,v_n]^T) = [\phi(v_1),\dots,\phi(v_n)]^T.
\end{equation}
The intermediate products 
\begin{equation}
    \mb \alpha^\ell( \mb x ) = \phi\Bigl( \mb W^{\ell} \phi\Bigl( \dots \phi \Bigl( \mb W^1 \mb x + \mb b^1 \Bigr) \dots + \mb b^{\ell} \Bigr)\Bigr)
\end{equation}
are sometimes referred to as {\em features} \cite{lecun2015deep}. In many situations, it is useful to ``pool'' features -- this is especially useful for data with spatial or temporal structure; combining spatially adjacent features in a nonlinear fashion renders the decision more stable with respect to deformations of the input \cite{scherer2010evaluation}. 
For example, {\em maximum pooling} takes the maximum of adjacent features. In our notation, we can denote this operation by $\rho^\ell$ and write 
\begin{equation}
    \mb \alpha^{\ell}(\mb x) = \rho^{\ell} \phi\left( \mb W^{\ell} \mb \alpha^{\ell-1}(\mb x) + \mb b^{\ell} \right),
\end{equation}
where the concise notation $\rho^{\ell}$ suppresses certain details about which features are combined. For clarity, we summarize this discussion in the following mathematical definition: 
\begin{definition}[Fully connected neural network] \label{def:fcnn} A fully connected neural network (FCNN) with {\bf feature dimensions} $n^0,\dots,n^L$, {\bf pre-activation dimensions} $m^1,\dots,m^L$, {\bf parameters}
\begin{eqnarray}
    \mb \theta &=& \Bigl( \, \mb W^{L} \in \bb R^{m^L \times n^{L-1}}, \dots, \mb W^1 \in \bb R^{m^1 \times n^0},  \nonumber \\
    && \qquad \mb b^{L} \in \bb R^{m^L}, \dots, \mb b^1 \in \bb R^{m^1} \, \Bigr),
\end{eqnarray}
{\bf activation function} $\phi : \bb R \to \bb R$ (extended to vector inputs by applying it elementwise), 
and {\bf pooling operations} $\rho^\ell : \bb R^{m^{\ell} \to n^{\ell}}$ given by
\begin{equation}
[\rho^\ell]_i(\mb v) = \max_{j \in I_i^\ell} \mb v_j,
\end{equation}
with $I_1^\ell, \dots, I_{n^\ell}^\ell$ being disjoint subsets of $[m^\ell]$,
is a mapping $f_{\mb \theta} : \bb R^n \to \bb R^d$ defined inductively as $f_{\mb \theta}(\mb x) = \alpha^L(\mb x)$ by setting $\mb \alpha^0(\mb x) = \mb x$, and
\begin{equation}
\mb \alpha^\ell(\mb x) = \rho^{\ell}\phi( \mb W^{\ell} \mb \alpha^{\ell-1}(\mb x) + \mb b^{\ell} ), \quad \ell = 1, \dots, L.
\end{equation}
\end{definition}

When discussing neural networks, it is conventional to distinguish between the  {\em network architecture}, which consists of the choices of feature dimensions $n^\ell, m^\ell$, activation function $\phi$, and pooling operators $\rho^\ell$, and the {\em network parameters} $\mb \theta$. Although we have stated a general definition, in specific architectures, the activation function $\phi$ and/or the pooling operators $\rho^\ell$ can be chosen to be trivial ($\phi(t) = t$ and/or $\rho^\ell(\mb v) = \mb v$).

{\bf Architectures.} Neural networks are flexible function approximators \cite{bengio2017deep}: universal approximation theorems indicate that {\em nonlinear}  neural networks (with non-polynomial activation $\phi$) can accurately approximate any continuous function, as long as the network is sufficiently deep and/or wide \cite{cybenko1989approximation,leshno1993multilayer,lu2017expressive}. There is a growing body of empirical and theoretical evidence showing that (relatively small) neural networks can learn relatively smooth functions over low-dimensional submanifolds of $\bb R^{n}$ with a complexity that is proportional to the manifold dimension, which in our problem equals the number of parameters in the parameterization $\mb \gamma \mapsto \mb s_{\mb \gamma}$ \cite{buchanan2020deep}.

Beyond these general considerations, there are  scenarios in which the nature of the task dictates specific architectural choices. For example, in the field of inverse problems, neural network architectures can be generated by interpreting various optimization methods as taking on the structure in Definition \ref{def:fcnn} \cite{gregor2010learning}. Our proposals will have a similar spirit, since they will interpret an existing method (matched filtering) as a particular instance of Definition \ref{def:fcnn}. 

Finally, a major architectural choice is whether to enforce additional structure on the matrices $\mb W^\ell$. When the input $\mb x$ is a time series, it is natural to structure the linear maps $\mb \alpha \mapsto \mb W \mb \alpha$ to be time-invariant, i.e., to be convolution operators. To exhibit the equivalence between matched filtering and neural networks in the simplest possible setting, here we train our networks on injections whose starting time is fixed, and focus on fully connected neural networks (not enforcing convolutional structure). 

In deployment, the input data is a time series, and astrophysical signals can occur at any time. In this setting, the matched filtering rule is applied in a sliding fashion. Similarly, the neural networks proposed here can be also deployed in a sliding fashion, which effectively converts them to particular convolutional networks. Both the equivalence between matched filtering and particular neural networks and the potential advantages of neural networks carry over to this setting. 

{\bf Parameters.} There are various approaches to choosing the network parameters $\mb \theta$. The dominant approach is to learn these parameters by optimization on data: one chooses initial parameters at random (with appropriate variance to ensure stability), and then iteratively adjusts them to best fit a given set of ``training data''. However, it is also possible in some scenarios to either (i) simply choose the weights at random, or (ii) to generate the weights analytically, either by connecting the network architecture to existing structures/algorithms \cite{gregor2010learning} or from harmonic analysis considerations \cite{bruna2013invariant}. There are approaches that lie in between purely data-driven and purely analytical approaches to choosing $\mb \theta$. For example, it is possible generate initial weights analytically, and then tune them on training data. This hybrid approach achieves excellent performance on a number of inverse problems in imaging (super-resolution \cite{wang2015deep}, magnetic resonance image reconstruction \cite{sun2016deep} etc.). 

In the following sections, we will follow this approach: we will give two ways of interpreting the matched filtering decision rule \eqref{eqn:mf-finite} as a fully connected neural network, by making specific (analytical) choices of the architecture and parameters. These analytically chosen parameters can then be used as an initialization for learning on data. We will also see that in addition to this closed-form construction for equivalence, neural network models can be further trained on data to achieve improved performance.

\subsection{Matched Filtering as a Shallow Neural Network}

In the language of the previous section, it is not hard to express the decision statistic \eqref{eqn:mf-finite} of matched filtering as a specific fully connected neural network with one layer ($L=1$). Writing 
\begin{eqnarray}
    \rho^1(\mb z) &=& \max_i z_i,  \\
    \phi(t) &=& t, \\
    \mb W^1 &=& \left[ \begin{array}{c} \mb s_{\mb \gamma_1}^*  \\ \mb s_{\mb \gamma_2}^* \\ \vdots \\ \mb s_{\mb \gamma_k}^* \end{array} \right] \in \bb R^{k \times n}, \label{eqn:shallow-1} \\
    \mb b^1 &=& \mb 0, \label{eqn:shallow-2} 
\end{eqnarray}
(\footnote{We note that the representation is not unique, and can be subject to shift and scale to produce essentially the same decision rule. Specifically, $\mb b^1$ can be identity vector times a constant (including zero) and $\mb W^1$ can be scaled by an arbitrary positive constant. However we choose the form given here for simplicity.})
we have 
\begin{equation} \label{eqn:equiv}
\max_{i} \innerprod{ \mb s_{\mb \gamma_i} }{ \mb x } = \rho^1 \phi \left( \mb W^1 \mb x + \mb b^1 \right). 
\end{equation}
In words, the features produced by this neural network correspond to the correlations of the input with the templates $\mb s_{\mb \gamma_1}, \dots,\mb s_{\mb \gamma_k}$. FIG \ref{fig:shallow-net} illustrates this (simple) architecture, which we label \shallownet. 
\begin{figure}[ht]
    \centering
    \includegraphics[width=3in]{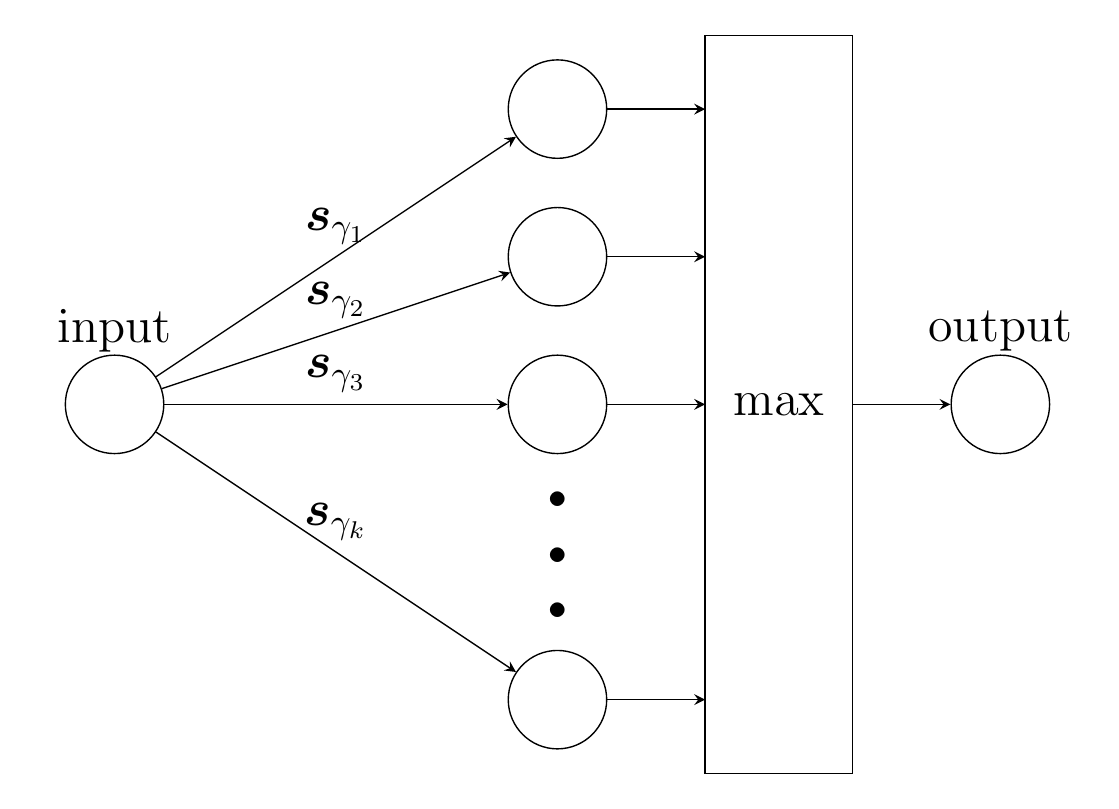}
    \caption{Illustration of \shallownet. Bias terms are omitted in the illustration. (We note that for more complex networks arbitrary pooling operations can replace the ``max'' box.)}
    \label{fig:shallow-net}
\end{figure}

Where needed below, we refer to the input-output relationship implemented by this architecture as  
\begin{equation}
f_{\text{\shallownet},\mb \theta}(\mb x),
\end{equation}
where $\mb \theta = (\mb W^1, \mb b^1)$ represent the weights and biases. When these are chosen as in \eqref{eqn:shallow-1}-\eqref{eqn:shallow-2}, \shallownet\ implements the matched filtering decision rule. We note that these weights can be constructed analytically based on the given templates. 

By learning the weights $\mb W^1$ and biases $\mb b^1$ from examples, we can further adapt this network to implement a more general family of decision rules, beyond matched filtering \eqref{eqn:mf-finite} with templates $\mb s_{\mb \gamma}$. Nevertheless, there are limitations to this architecture. Notice that in \shallownet\ there is only one layer of affine operations,
and so this architecture does not satisfy the dictates of the universal approximation theorem \cite{leshno1993multilayer, kidger2020universal}. 

More geometrically, we can notice that the decision rule associated with \shallownet \ is a maximum of affine functions. This means that for any choice of $\mb W^0$ and $\mb b^0$, the decision boundary is the boundary of a convex set. This property is also true for matched filtering, which shares exactly the same form. An illustration of this property is shown in FIG \ref{fig:convex-boundary}.
\begin{figure}[ht]
    \centering
    \includegraphics[width=2in, trim={0 .1in 0 0}]{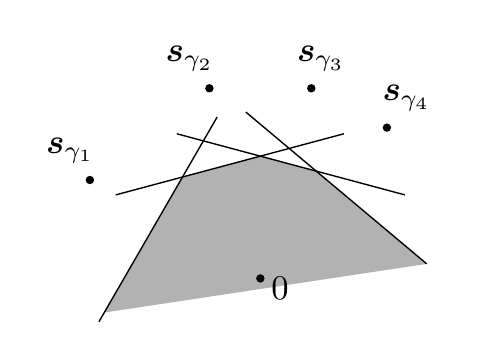}
    \caption{The set of points classified as noise by matched filtering and \shallownet\ is always a convex set.}
    \label{fig:convex-boundary}
\end{figure}

{\em How restrictive is this limitation?} In the context of parametric detection, this depends largely on the noise distribution. If the noise is Gaussian, the optimal decision boundary is itself the boundary of a convex set:

\begin{proposition} \label{prop:convex-boundary} Suppose that the noise $\mb z \sim \mc N(\mb 0, \sigma^2\mb I)$. Then for any significance level $\alpha$, the optimal (Neyman-Pearson) decision region 
\begin{equation}
    \{ \mb x \mid \lambda(\mb x) \le \tau \}
\end{equation}
is a convex subset of $\bb R^n$, where $\tau$ is a constant determined by the significance level $\alpha$.
\end{proposition}

\begin{proof} 
Please see Appendix. 
\end{proof}

However, for general non-Gaussian distributions, the optimal decision region is often nonconvex. We illustrate this result in FIG \ref{fig:noises}. In fact, this suggests an intrinsic structural limitation of matched filtering and similar architectures. Since in reality the noise distribution is not perfectly Gaussian, we cannot expect the optimal decision region to be convex, and hence the matched filtering structure is unable to approach the performance of the likelihood ratio test with arbitrary precision, even if any number of templates (including ones outside the original signal space) are allowed. In such cases, we can benefit from using a more flexible architecture, which we now introduce.
\begin{figure*}[ht]
    \centering
    \includegraphics[width=7in, trim={0 .1in 0 .1in}]{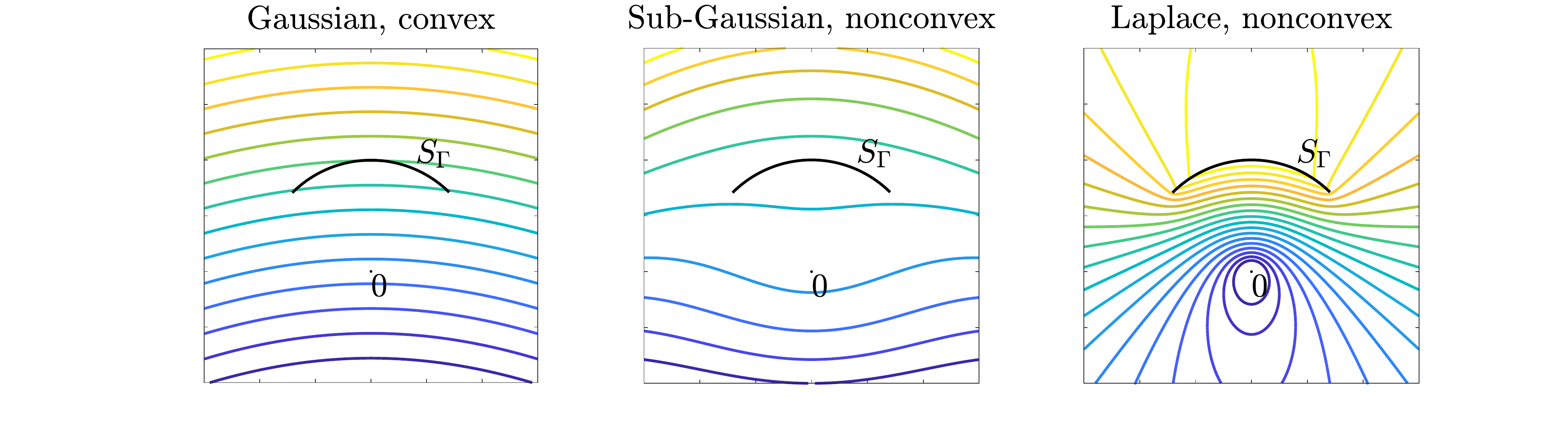}
    \caption{Contours of log likelihood ratio with various noise distributions, and whether the optimal decision regions with $\delta=0$ is always convex. Yellow represents larger values and blue represents lower values. From left to right: (1) Gaussian distribution, convex; (2) Sub-Gaussian distribution $\rhonoise(\mb x)\propto \exp(-C\|\mb x\|^3)$, not necessarily convex; (3) Laplace distribution, not necessarily convex. 
    }
    \label{fig:noises}
\end{figure*}

\subsection{Matched Filtering as a Deep Neural Network}

We describe an alternative way of expressing template matching as a neural network, which leads to deep, nonlinear architectures that are more flexible than \shallownet. We label this structure \deepnet. In this architecture, we do not compute the maximum in a straightforward way using pooling. Instead, we propose an alternative architecture which is more flexible, and can approximate a wider class of functions. In particular, we will no longer be restricted to implementing decision boundaries that are boundaries of convex sets, allowing us to handle scenarios with non-Gaussian noise. An illustration of this \deepnet \ is shown in FIG \ref{fig:deep-net}. 

\begin{figure}[ht]
    \centering
    \includegraphics[width=3.5in]{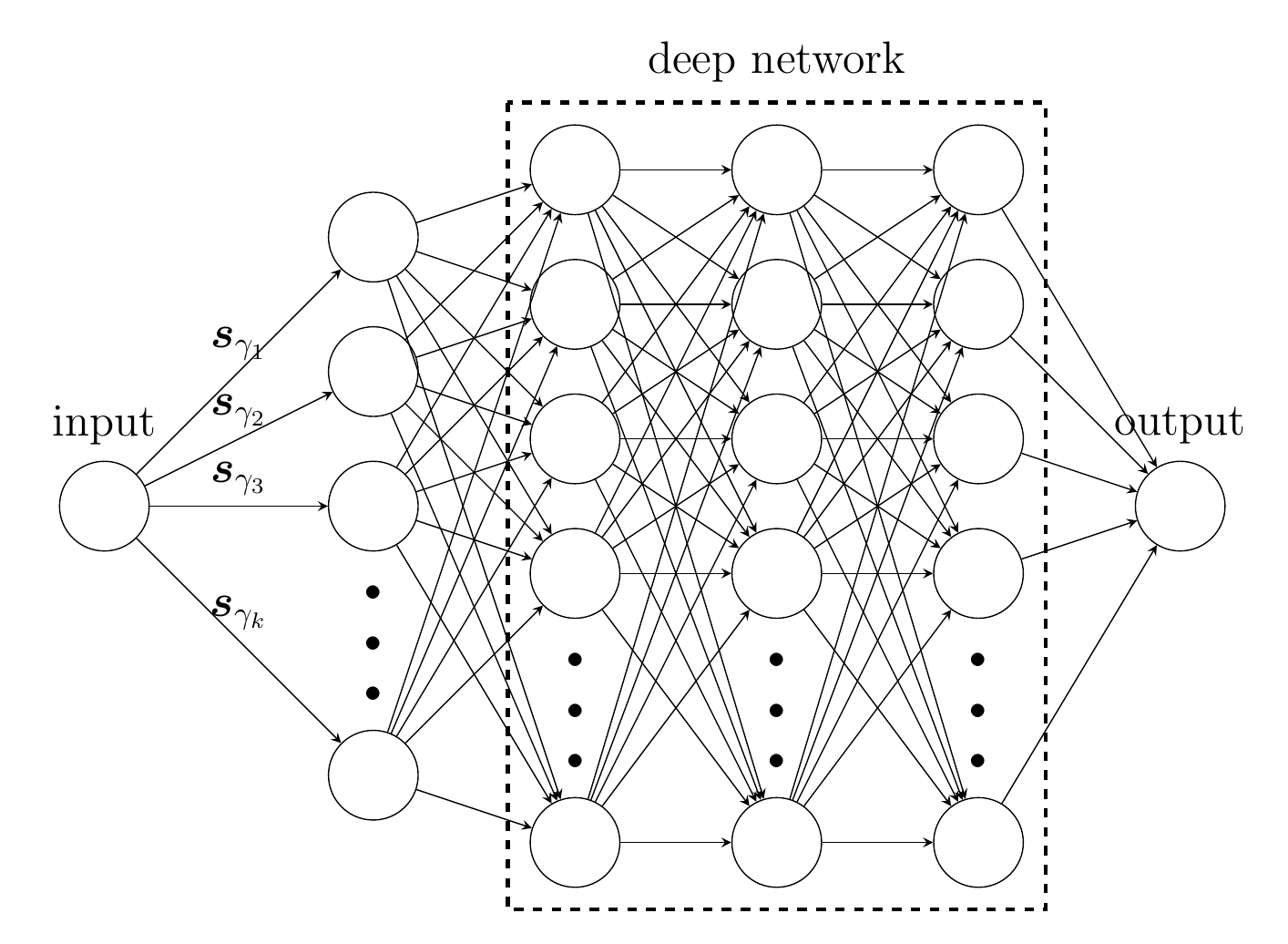}
    \caption{Illustration of \deepnet. Bias terms are omitted in the illustration. This network structure is obtained by replacing the $\max$ module in matched filtering (as in FIG \ref{fig:shallow-net}) with a deep network.}
    \label{fig:deep-net}
\end{figure}

Our construction is based on the rectified linear unit (ReLU) nonlinearity:
\begin{equation}
    \phi(t) = \max( t, 0 ). 
\end{equation}
This is arguably the most commonly used nonlinearity function in modern deep learning. 

The matched filtering decision rule takes the maximum of a family of linear functions $\innerprod{\mb s_{\mb \gamma_i}}{\mb x}$. Instead of simply ``pooling'' these functions as in the previous section, we implement the maximum operation using compositions of ReLUs and linear operations. In particular, observe that the maximum of two numbers can be written as a linear combination of 3 ReLU units: \begin{equation}
    \max(a,b) = b + \phi(a-b) = \phi(b) - \phi(-b) + \phi(a-b).
\end{equation}
The basic idea is to create a hierarchical structure of such 3-ReLU-units, each of which takes a pairwise maximum of its inputs. Our \deepnet\ construction will perform convolutions with the templates $\mb s_{\mb \gamma_i}$, followed by this hierarchical structure for computing the maximum.

\begin{figure}[ht]
    \centering
    \includegraphics[width=3.5in]{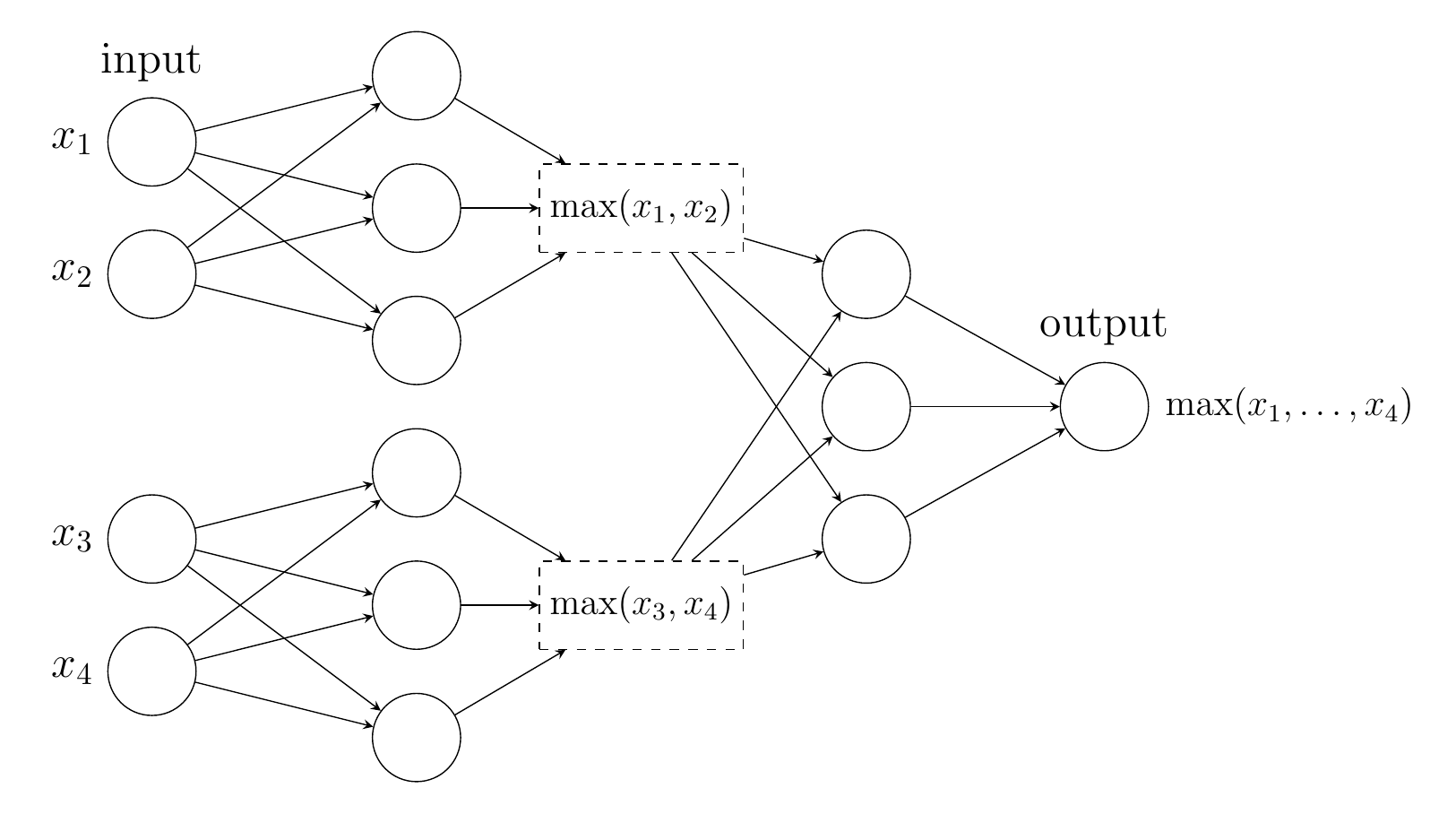}
    \caption{Illustration of implementing max with a ReLU network. The dashed boxes in the middle are not actual nodes in the network, but ``imaginary'' nodes to facilitate construction.}
    \label{fig:pairwise-max}
\end{figure}

FIG \ref{fig:pairwise-max} illustrates this hierarchical structure for the particular example of four inputs. The network in FIG \ref{fig:pairwise-max}
can be expressed as a ReLU network, with sparse weight matrices $\mb W^\ell \ (\ell=0,1,2)$ for the layers respectively: 
\begin{eqnarray}
    \mb W^0 &=& \begin{bmatrix}
        0 & 1 \\
        0 & -1 \\
        1 & -1
    \end{bmatrix}, \quad
    \mb W^2 = \begin{bmatrix}
        1 & -1 & 1
    \end{bmatrix}, \\ 
    \mb W^1 &=& \mb W^0 \otimes \mb W^2 = \begin{bmatrix}
        0 & 0 & 0 & 1 & -1 & 1 \\
        0 & 0 & 0 & -1 & 1 & -1 \\
        1 &  -1 & 1 & -1 & 1 & -1
    \end{bmatrix}.
\end{eqnarray}
Generalizing this construction, we obtain a network that takes the maximum of $k$ numbers, using $\lceil \log_2  k \rceil + 1$ layers. 

While the example above delineates a precise form of the ReLU network, this approach can in fact be made flexible. To ensure that the network output is indeed the maximum of the $k$ inputs, we must ensure that at each layer, each feature participates in at least one of the pairwise max operations. This means that at layer $\ell$, we must have at least $k / 2^\ell$ features. However, we are free to add more intermediate features, with additional (redundant) max operations. This does not change the output of the network, but it affords additional flexibility when we attempt to train the network on data. 
In particular, this allows the construction of arbitrarily wide or deep ReLU networks, and can therefore approximate any regular continuous function \cite{leshno1993multilayer, kidger2020universal}.

There is also a degree of freedom in choosing which features participate in each pairwise maximum operation, which could be chosen in various ways. In our implementation we use the following way to pair up the nodes in layer $l$ for pairwise maximum operations that get to layer $l+1$. Assume layer $l$ contains $2p$ nodes. First pair up the nodes with consecutive indices, namely pair up node $2i-1$ with node $2i$ for $i=1,\dots,p$. This ensures that each node is covered by at least one maximum operation. After that, for each leftover node in layer $l+1$, we establish the corresponding pair in layer $l$ by choosing the nodes at random in layer $l$. In the following, we label this network \deepnet. We emphasize for clarity that the nodes between consecutive layers are fully connected in the neural network; however, the weights not associated with pairwise maximum operations are all initialized to zero. Below, where needed we refer to the decision rule associated with this network as 
\begin{equation}
    f_{\deepnet,\mb \theta}(\mb x),
\end{equation}
where $\mb \theta$ represent the collection of all weights and biases. The above discussion again gives a recipe for choosing these weights analytically such that the decision rule for \deepnet\ coincides with the matched filtering rule.

In contrast to \shallownet, \deepnet \  is a more flexible architecture. In particular, this architecture satisfies the dictates of the universal approximation theorem. Geometrically, it is not restricted to convex decision regions, which makes it capable of achieving optimal decision boundaries even when the noise is heavy-tailed or has other non-ideal properties. 

\subsection{Equivalence of Matched Filtering and Neural Networks}

We have demonstrated by construction the following claim: 
\begin{quote}
{\em Given any collection of templates $\mb s_{\mb \gamma_1}, \dots, \mb s_{\mb \gamma_k}$ (for any $k \ge 1$), one can analytically determine weights $\mb \theta_s$, $\mb \theta_d$ such that 
\begin{eqnarray}
    f_{\text{\shallownet},\mb \theta_s}(\mb x) &=& \max_{i =1 \dots k} \innerprod{ \mb s_{\mb \gamma_i} }{ \mb x }  \\
    f_{\text{\deepnet}, \mb \theta_d}(\mb x) &=& \max_{i =1 \dots k} \innerprod{ \mb s_{\mb \gamma_i} }{ \mb x }  
\end{eqnarray}
for all $\mb x \in \mathbb{R}^n$. }
\end{quote}
We emphasize the complete generality of this claim: it holds for any number and choice of templates. Moreover, it does not depend on training: the networks can be constructed analytically to implement the matched filtering rule. Nevertheless, we will see in the next section that they can be further adapted based on observed data to strictly outperform matched filtering, in terms of the Neyman-Pearson criterion. 

The equivalence between matched filtering and particular neural networks has an additional conceptual advantage: it allows for a clear comparison of the resource complexity of different search methods, in terms of storage and computation. This is valuable because different methods may cut out very different tradeoffs between complexity and accuracy/performance. Neural network implementations of matched filtering can be viewed as ``complexity standard candles'' against which the performance of more sophisticated networks can be measured.
In particular, the complexity of a neural network model may be quantified by the total number of nodes (neurons) in the network, which approximately characterizes the number of elementary operations performed for evaluating an input instance \cite{orponen1994computational, bianchini2014complexity}. We will look for the most appropriate measure of complexity for this problem, and provide detailed analysis in future studies.

\section{Training to Approach Statistical Optimality}

In the previous section, we gave two ways of analytically constructing neural networks that reproduce the matched filtering decision rule, and hence exhibit exactly the same performance as matched filtering. The major advantage of this interpretation of matched filtering is that the resulting model can be further trained on sample data to improve its statistical performance or adapt it to handle non-Gaussian noise distributions, or in other words ``standing on the shoulder of giants''. 
In a typical neural network training problem, we have access to labelled samples 
\begin{equation}
    (\mb x_1, y_1), \dots, (\mb x_N, y_N),
\end{equation}
each of which consists of an observation $\mb x_i \in \bb R^n$ and a corresponding label $y_i \in \{ 0, 1 \}$, which indicates whether $\mb x_i$ contains a noisy signal ($y_i = 1$) or noise only ($y_i = 0$). To date, we have only a moderate number of confirmed gravitational wave detections, and hence have far more negative examples than positive examples. We address this issue by generating our positive training examples by injecting synthetic waveforms into (real) LIGO noise strains. Below, we describe two different training schemes, motivated by the Neyman-Pearson and minimax criteria, which leverage this data to perform training of the neural networks.

{\bf Training for Neyman-Pearson.} In this setting, we assume that the prior $\nu$ is known, and generate positive examples by first sampling $\mb \gamma_i \sim \nu$, and setting $\mb x_i = \mb s_{\mb \gamma_i} + \mb z_i$, where $\mb z_i$ is observed LIGO noise strain. We solve the following optimization problem:
\begin{equation}
    \min_{\mb \theta} \mc R_N( f_{\mb \theta} ) := \frac{1}{N} \sum_{i=1}^N \ell\Bigl( f_{\mb \theta}(\mb x_i), y_i \Bigr). 
    \label{eqn:risk-minimization}
\end{equation}
Here, the {\em loss function} $\ell( \hat{y}, y )$ measures the misfit between the predicted label $\hat{y}$ and the true label $y$. Typical choices include the square loss $(\hat{y} - y)^2$ and the logistic loss
\begin{equation}
y \log(f_{\mathrm{sigmoid}}(\hat{y})) + (1 - y) \log(1 - f_{\mathrm{sigmoid}}(\hat{y})).
\label{eqn:logistic-loss}
\end{equation}
where $f_{\mathrm{sigmoid}}(\cdot)$ denotes the logistic/sigmoid function:
\begin{equation}
    f_{\mathrm{sigmoid}}(x) = \frac{1}{1+\exp(-x)}
\end{equation}

{\em Is this training strategy compatible with the Neyman-Pearson criterion?} The following proposition answers this question in the affirmative. Consider the following setup: training data $(\mb x_i, y_i)$ are generated independently at random, by setting $y_i = 1$ with probability $p \in (0,1)$ and choosing $\mb x_i = \mb s_{\mb \gamma_i} + \mb z_i$ when $y_i = 1$ and $\mb x_i = \mb z_i$ when $\mb y_i = 0$, with $\mb \gamma_i \sim \nu$, and $\mb z_i \sim \rho_{\mr{noise}}$. Let \begin{equation}
    \mc R_\infty ( f ) = \bb E_{(\mb x,y)} \ell( f(\mb x), y ). 
\end{equation}
This represents the large-sample limit of $\mc R_N$: as $N \to \infty$, $\mc R_N(f) \to \mc R_\infty(f)$. 
The following proposition shows that the population risk $\mc R_\infty$ is minimized by (a monotone function of) the likelihood ratio $\lambda$:

\begin{proposition} \label{prop:risk-minimization} Suppose that for any $y = 0,1$, the loss $\ell(\wh{y},y)$ is a strictly convex differentiable function of $\wh{y}$ that is minimized at $\wh{y}=y$. \footnote{In fact it is straightforward to show that the conclusion of Proposition \ref{prop:risk-minimization} holds for more general classes of loss functions, including the logistic loss.} 
Then the unique optimal solution $f_\star$ to the (functional) optimization problem 
\begin{equation} 
\min_f \mc R_\infty( f ) 
\end{equation}
is a strictly increasing function of the likelihood ratio $\lambda$:
\begin{equation}
    f_\star(\mb x ) = g( \lambda(\mb x) ),
\end{equation}
where $g$ is a strictly increasing function that depends on $\ell$.
\end{proposition} 
\begin{proof} 
Please see Appendix. 
\end{proof} 

This result can be interpreted as saying: {\em ``a sufficiently flexible classifier, trained on a sufficiently large dataset will produce the optimal decision rule.''}  Hence, training to minimize the empirical risk $\mc R_N(f_{\mb \theta})$ is compatible with the Neyman-Pearson criterion.

While this is a promising observation, we should keep in mind a number of remaining issues: How much data is required? What are effective approaches to minimizing the empirical risk $\mc R_N$? In the next section we investigate these questions experimentally.

{\bf Training for Minimax.} 
In this setting, we do not assume any prior, and aim to minimize the worst false negative rate using the formulation in \eqref{eqn:minimax}. We convert the constrained problem \eqref{eqn:minimax} to an equivalent unconstrained problem, 
\begin{equation}
    \min_{\delta} \; \max_{\mb \gamma \in \Gamma} \fnr_{\mb \gamma}   + c\cdot \fpr,
\end{equation}
where $c$ is a constant that depends on $\alpha$.  For tractability, we will fix $c$ at a constant value to obtain a concrete optimization objective, and here we fix $c=1$. In actual deployment where a target significance level $\alpha$ is specified, we can also choose $c$ at the level that corresponds to the specified $\alpha$.
Also, we sample the parameter space $\Gamma$ at points $\{\mb\gamma_i\}_{i=1}^N$. Since $\fpr$ does not depend on $\mb\gamma$, it can be moved inside the maximization. Therefore, the minimax optimization problem can be transformed into
\begin{equation}
    \min_{\delta} \; \max_{i=1,\dots,N} \left( \fnr_{\mb \gamma_i}  + \fpr \right).
\end{equation}

This suggests a natural approach to training under the minimax criterion using first-order optimization methods. At each iteration, we estimate $\fpr$ and $\fnr_{\mb\gamma_i}$ for each $i=1,\dots,N$, and choose $i_*$ with the highest $\fnr_{\mb\gamma_i}$. We then aim to reduce $\fnr_{\mb\gamma_i}+\fpr$, which can be estimated by using a sample dataset $\{(\mb x_i, y_i)\}_{i=1}^N$ as
\begin{equation}
    \frac{1}{N} \sum_{i=1}^N \mb 1 \Big[f_{\mb\theta}(\mb x_i) \neq y_i\Big],
\end{equation}
where in the dataset all $\mb x_i$ with corresponding $y_i=1$ are generated specifically with signal parameter $\mb\gamma_{i_*}$, and half of data pairs in the dataset have $y_i=0$. Finally, it is customary in optimization to replace the non-differentiable 0-1 loss with a smooth loss function $\ell$, and hence we get the following risk minimization objective:
\begin{equation}
    \frac{1}{N} \sum_{i=1}^N \ell \Big(f_{\mb\theta}(\mb x_i), y_i\Big).
\end{equation}
This expression is similar to \eqref{eqn:risk-minimization}, but the difference is that all positive data in the dataset here are associated with signal parameters $\mb\gamma_{i_*}$.

\section{Simulations and Experiments}

\subsection{Data Generation}

Data-driven methods such as neural networks typically require a large amount of data for training. The question of data sufficiency is especially acute in gravitational wave astronomy: we have only a moderate number of confirmed detections to date. We address this issue by generating our positive training examples by injecting synthetic waveforms into LIGO noise strains~\cite{ABBOTT2021100658}, which we elaborate below.

For LIGO noise data, we use the L1 strain from LIGO O2 run between August 1 and August 25, 2017, with ANALYSIS{\textunderscore}READY segments only. The announced confident detections GW170809, GW170814, GW170817, GW170818 and GW170823 are removed from the strain, such that the data is at least $300$ seconds away from these events. We used a total of $338$ frame files each of $4096$ seconds long, namely a total of $384.57$ hours. The strain data is downsampled from the original $4096$Hz to $2048$Hz for processing efficiency. The downsampled L1 strain data is divided into segments of length $0.6$ second, with each successive segment overlapping $50\%$ of the previous segment. 

We generate synthetic gravitational wave signals using 
PyCBC~\citep{PyCBCSoft,2018PhRvD..98b4050N,2016CQGra..33u5004U,2014PhRvD..90h2004D,2005PhRvD..71f2001A,2019PASP..131b4503B,2012PhRvD..85l2006A}, with the following parameters. \textit{Approximant:} IMRPhenomD. \textit{Mass range:} $40$ to $50$ $M_{\bigodot}$, uniformly distributed. \textit{Spin:} $0$. \textit{Sampling rate:} $2048$Hz. \textit{Low frequency cutoff:} $30$Hz. \textit{Coalescence phase:} 0. \textit{Polarization:} plus~\cite{2011gwpa.book.....C}. With this specified mass range, at least $99.5\%$ of the energy of the signal lies in an interval of length $0.3$ second after preprocessing. 
We note that although the templates are not chosen uniformly in actual LIGO deployment~\cite{1996PhRvD..53.6749O,1996PhRvD..53.3033B,1999PhRvD..60b2002O,1998PhRvD..57.2101B,2009PhRvD..79j4017M}, we make this choice here due to simplicity, and also the fact that the large number of templates make up for the possibly suboptimal choice of templates.
 
The above data is used to generate training and test datasets of positive and negative labelled data as follows. We divide the collection of downsampled strain segments randomly into training and test sets, ensuring that no training segment overlaps a test segment. Within the training and test sets, we generate both positive and negative examples. The negative examples contain only the strain data. 
For the positive examples, we inject waveforms into the noise segments by aligning the peak of the waveforms at the $90\%$ location of the center $0.3$s, namely at the location of $0.42$s within the entire segment of $0.6$s. This choice was made as it safely covers the injected waveforms.
The amplitude of the injection is set such that after filtering and whitening (to be described below), the resulting signal-to-noise ratio (SNR) is constant. For the experiment, the size of the training and test datasets are respectively $2.62$ million and $2$ million segments.

We preprocess all training and test data, by applying an FIR bandpass filter with cutoff frequencies $30$Hz and $400$Hz, whitening  using a power spectral density estimated from the L1 strain data, and finally truncating to keep only the center $0.3$ second ($614$ samples).

\subsection{Matched Filtering Configuration}

We first need to determine the necessary number of templates to use in matched filtering, given the space of parameters. We set 10, 100, 1000 and 10000 as the candidate numbers of templates. For each candidate number, we independently repeat the following process 30 times: randomly choose the specified number of pairs of parameters uniformly from $[40,50]\times [40,50]$, generate waveforms according to these parameters, preprocess (bandpass, whiten and truncate) as described above, and then normalize to equal power. This produces the templates for a matched filtering model. We evaluate the model on the test dataset to obtain an ROC curve. For each candidate number of templates and for each value of FPR, we take the lowest FNR outcome among the 30 independent runs. This is used to approximately represent the best performance achievable with a given number of templates. 

The result is shown in FIG \ref{fig:choose-mf}. We see that the best performance of matched filtering in this setting starts to saturate at approximately 1000 templates, and the best performance with 1000 templates is almost identical to the that with 10000 templates. Therefore, we choose the best performance of matched filtering with 10000 templates, namely the bright blue curve, as the performance curve of the matched filtering method in this setting, against which we will be comparing our neural network method.

\begin{figure}[ht]
    \centering
    \includegraphics[width=3in]{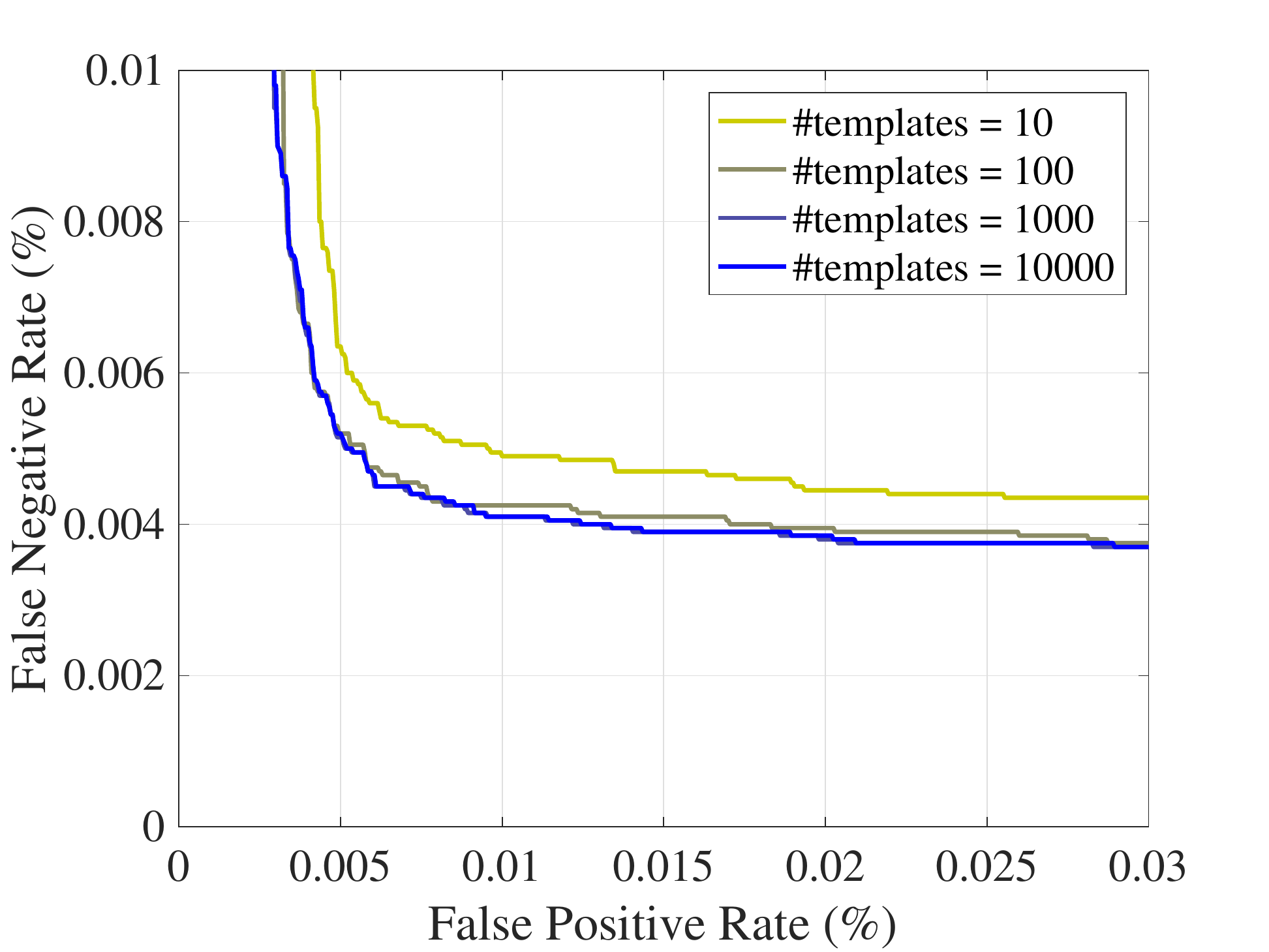}
    \caption{The best performance of matched filtering with given number of templates across 30 independent runs. The performance starts to saturate above 1000 templates.}
    \label{fig:choose-mf}
\end{figure}

\subsection{Neural Network Configuration}

To initialize the templates of the neural network models for both \shallownet\ and \deepnet, we generate $1000$ random waveforms from a uniform distribution over the same parameter range, subject to the same preprocessing and normalization process as done in matched filtering.

For the \deepnet\ architecture, in addition to the 1000 initialized templates, we also need to specify the number of layers and the feature dimension of each layer. In the experiment we choose $L=17$ and \begin{align*}
    (n_1, n_2,\dots,n_L) &= (1000, 1800, 1200, 720, 480, 300, 180,  \\
    &\qquad  120, 90, 60, 36, 24, 18, 12, 6, 3, 1).
\end{align*}
Here these feature dimensions $n_l$ are chosen arbitrarily so long as they satisfy $n_2\ge \frac{3}{2}n_1$, $n_\ell\ge \frac{1}{2}n_{l-1}$ for all $3\le\ell\le L-1$, $n_{L-1}=3$, $n_L=1$, and that $n_2,\dots,n_{L-2}$ are all divisible by 6 (which facilitates construction using our proposed initialization scheme).

For minimax training, in order to search the parameter space for the worst performance, we sample the parameter space $[40,50]\times [40,50]$ of $(m_1,m_2)$ using a square grid sampler with interval 0.5. After discarding equivalent samples due to the symmetry between $m_1$ and $m_2$, there are in total $231$ samples in the parameter space.

For the optimization parameters of the neural network, we train the network using logistic loss, the Adam optimizer \cite{kingma2014adam}, and a constant learning rate of $10^{-5}$.

\subsection{Simulation Results}

{\bf Performance under minimax.} In this experiment we perform injections such that SNR is 5, and only for the \shallownet\ model. While this SNR value is smaller than the range of meaningful observed events, we choose this value for the simplicity of exposition and reduction of training time, since the training procedure for minimax criterion is rather computationally heavy. Similar results should hold at higher SNR values. FIG \ref{fig:better-minimax} plots the ROC curves for both matched filtering and \shallownet\ trained for minimax, measured in terms of both worst performance and the average performance over a uniform prior. We see that the trained neural network achieves better performance than matched filtering under minimax, while achieving approximately identical performance as matched filtering under Neyman-Pearson with a uniform prior. This is not surprising since the training process is designed to only optimize for the minimax criterion, and not the Neyman-Pearson criterion with uniform prior.

\begin{figure}[ht]
    \centering
    \includegraphics[width=3in]{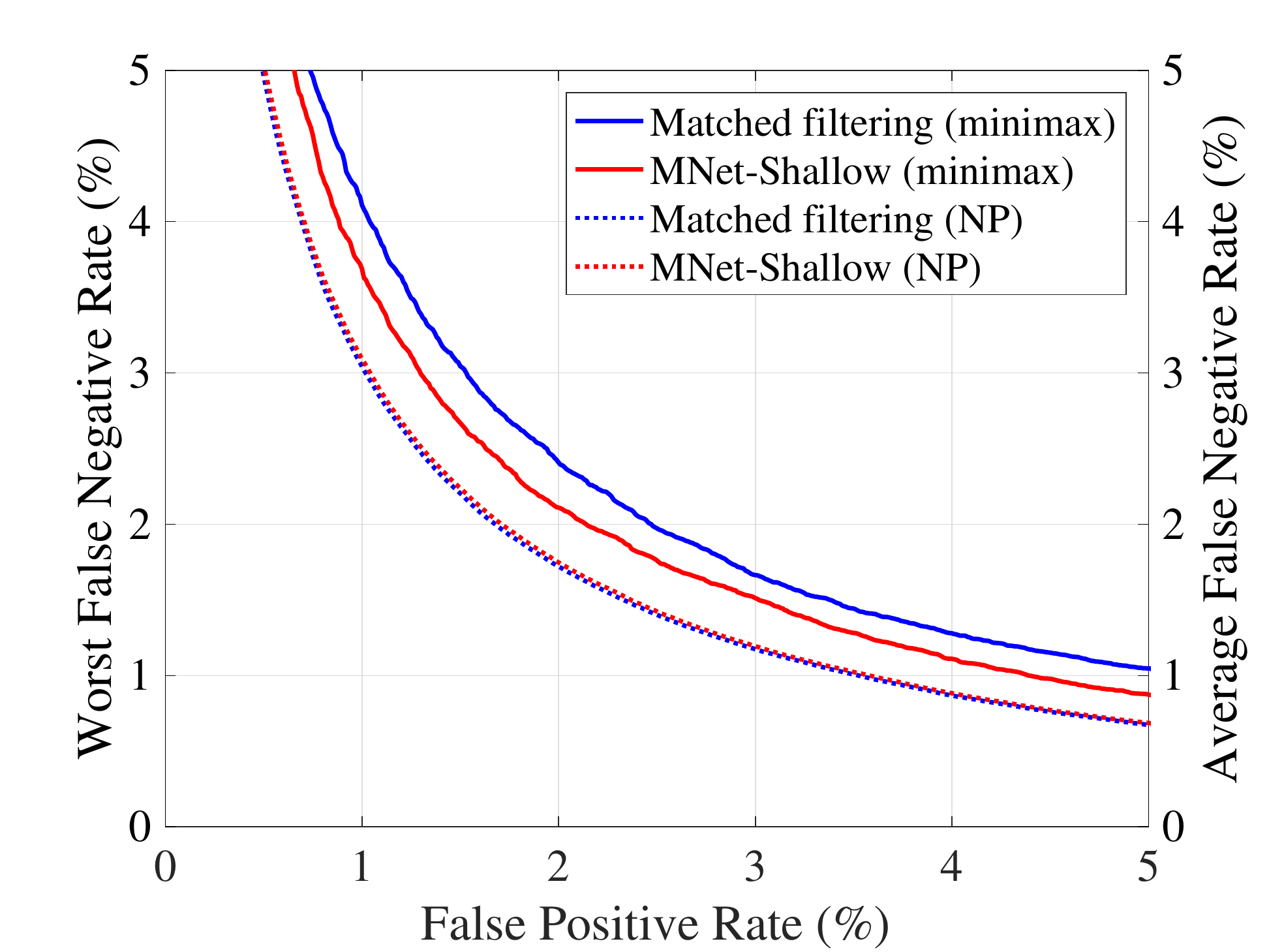}
    \caption{ROC curves of the trained shallow neural network and matched filtering. The solid curves correspond to the vertical axis on the left, and the dotted curves correspond to the vertical axis on the right. For both models we show both the worst (minimax) performance and the average performance under Neyman-Pearson (NP) setting with a uniform prior. The neural network with minimax training outperforms matched filtering in terms of the minimax criterion. The performance of the two models under NP is similar, which is reasonable since our optimization for the neural network was aimed for the minimax criterion only.}
    \label{fig:better-minimax}
\end{figure}

{\bf Performance under a uniform prior.} In this experiment we perform injections such that SNR is 9. Figure \ref{fig:better} plots the ROC curves for both formulations \shallownet\ and \deepnet\ trained for Neyman-Pearson, as well as that of matched filtering. As expected, the neural network models strictly improves over matched filtering. Moreover, the \deepnet\ architecture has a slight performance advantage over \shallownet.
The performance improvement of the trained models over matched filtering is especially remarkable with low FNR values, which is arguably the more important scenario for gravitational wave detection, since we can hardly afford to miss actual astrophysical events which are quite scarce.

\begin{figure*}[t]
    \centering
    \includegraphics[width=7in]{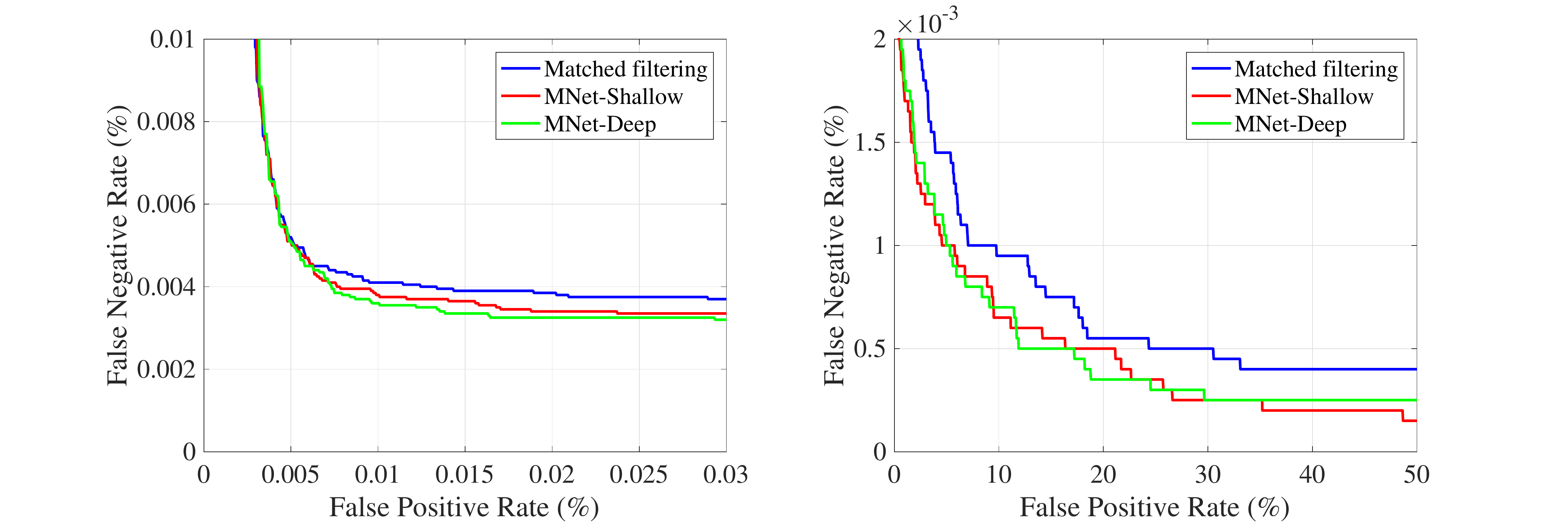}
    \caption{ROC curves of the trained \shallownet\ and \deepnet\ models compared with matched filtering. Left and right panels plot the same curves, but have different axis ranges to better show the contrast between the curves.
    }
    \label{fig:better}
\end{figure*}

\section{Discussion}

Our experiments demonstrate the potential of neural networks to outperform matched filtering, especially at low false negative rates. 
The flexibility of neural networks also enables this architecture to implement more general variations of matched filtering, such as with weights or aggregation functions different from the maximum. Neural networks have additional potential advantages: 
deep networks can adapt to unknown and/or non-Gaussian noise distributions. In addition, architectural ideas in deep networks such as pooling help to convey invariances that may be helpful in detecting some ``unknown unknowns'' that lie outside of the span of a pre-specified family of templates. This should be investigated in the future.

The proposed architectures can be adapted to time-varying noise distributions, by pre-training on very large collections of (synthetic) Gaussian noise and then adapting the pre-trained network using a smaller number of online examples. This kind of pre-training may also be helpful in deploying our methods across larger mass ranges, which require more training data. 

We note that it is, in some sense, unsurprising  that deep networks can exhibit advantages over matched filtering, since the former can be made arbitrarily complex, and can approximate essentially arbitrary functions. An important direction for future work is to study architectures that not only approach optimal statistical performance, but exhibit good {\em complexity-performance} tradeoffs. There are a number of concrete directions for achieving this -- in particular, the weight matrices learned by our Neyman-Pearson networks exhibit particular types of low-dimensional (low-rank and sparse) structure, which can be leveraged to reduce complexity. Interpreting matched filtering as a particular neural network facilitates the study of complexity-performance tradeoffs, since it allows these distinct methods to be studied in a unified framework. Another avenue for complexity reduction is to define and train very large (overparameterized) networks and then prune them to produce much smaller subnetworks with good performance. \deepnet\ is particularly promising in this regard, since this construction yields networks of arbitrary depth. 

One future possibility of the approach is to go beyond the fixed template banks that constrain the limited set of parameters taken into account. For example, to limit the size of the template bank, BH spins that are misaligned from the orbital angular momentum are not widely used yet. Also, due to the lack of available template banks, some astrophysically feasible scenarios receive relatively little attention, including eccentric binary merger template banks where every new template requires a computationally very expensive general-relativity simulation. Therefore generalized matched filtering needs to be investigated in this context, to measure its performance on signal classes that current templates don't cover. Additionally, training it with a sample of eccentric waveforms could enable the detection of other eccentric BBHs even with properties not covered by the limited simulation used for training. Exploring these scenarios are very important experiments for the future.

Another desirable goal is to allow matched filtering algorithms to run "coherently", treating the GW detectors worldwide as a single detector and analyzing data from multiple GW detectors together as a single data stream. The main difficulty is that the sky direction of the cosmic source is unknown, therefore there are many unknown time shifts among the detectors' data. Searching a large number of different combinations can be cost prohibitive with current approaches. It is important to experimentally investigate the ML extensions to matched filtering to measure the increased sensitivity due to the coherent framework.

Furthermore, experiments on the natural generalization of the approach where one does not aim to find the best matching waveform, but instead aims to estimate the parameters of the BBH system are needed. For example, instead of having the maximum reported, one could report the probability distribution over parameters. The difficulty here is that searches usually have much fewer parameters than what is used for parameter estimation. The performance of the ML framework in parameter estimation should be quantified in the future, even if it comes at the price of precision and is therefore only used as a first estimate.

\section{Conclusion}
\label{sec:conclusion}

In this paper, we highlighted the idea that matched filtering currently applied by LIGO is formally equivalent to a particular neural network, which can be defined analytically in closed form. We also modeled the LIGO gravitational wave search as the parametric signal detection problem, and illustrated the suboptimality of matched filtering regardless of whether a prior distribution on the parameter space is given. On the other hand, we proposed neural network architectures \shallownet\  and \deepnet, which are initialized to implement matched filtering exactly, and then trained on data for improved performance. In particular, we showed that when the prior distribution is known, the training process is aligned with the statistically optimal decision rule. Between the two proposed architectures, the former more closely resembles the architecture of matched filtering, while the latter has a more flexible architecture capable of dealing with a wider range of distributions. We conducted experiments using LIGO strain data from O2 and synthetic waveform injections, and showed that our trained network can achieve uniformly better performance than matched filtering both with or without a known prior, especially in scenarios where false negative rate is low.

Through this work, we seek to bridge the gap between data-driven methods such as deep learning and those detection methods currently in use in LIGO, and explore the possibility of incorporating them into the gravitational wave search of LIGO, as well as broader areas of scientific discovery. In the future work, we aim to explore the potentials of efficiency gains of neural networks over matched filtering, and also establish an end-to-end guarantee for the performance of the proposed framework.

\begin{acknowledgments}
We acknowledge computing resources from Columbia University's Shared Research Computing Facility project, which is supported by NIH Research Facility Improvement Grant 1G20RR030893-01, and associated funds from the New York State Empire State Development, Division of Science Technology and Innovation (NYSTAR) Contract C090171, both awarded April 15, 2010.

The authors are grateful for the LIGO Scientific Collaboration for the careful review of the paper. This paper is assigned a LIGO DCC number of LIGO-P2100086. 
The authors acknowledge the LIGO Laboratory and Scientific Collaboration for the detectors, data, and the game changing computing resources (National Science Foundation Grants PHY-0757058 and PHY-0823459). The authors would like to thank colleagues of the LIGO Scientific Collaboration and the Virgo Collaboration and Columbia University for their help and useful comments, in particular the CBC group, Andrew Williamson, Stefan Countryman, William Tse, Nicolas Beltran, Asif Mallik, Sireesh Gururaja, and Thomas Dent which we hereby gratefully acknowledge.
SM thanks David Spergel, Rainer Weiss, Rana Adhikari, and Kipp Canon for the motivating general discussions related to the role of machine learning and data analysis. 

This research has made use of data, software and/or web tools obtained from the Gravitational Wave Open Science Center~\cite{ABBOTT2021100658,GWOSC} (https://www.gw-openscience.org/ ), a service of LIGO Laboratory, the LIGO Scientific Collaboration and the Virgo Collaboration. LIGO Laboratory and Advanced LIGO are funded by the United States National Science Foundation (NSF) as well as the Science and Technology Facilities Council (STFC) of the United Kingdom, the Max-Planck-Society (MPS), and the State of Niedersachsen/Germany for support of the construction of Advanced LIGO and construction and operation of the GEO600 detector. Additional support for Advanced LIGO was provided by the Australian Research Council. Virgo is funded, through the European Gravitational Observatory (EGO), by the French Centre National de Recherche Scientifique (CNRS), the Italian Istituto Nazionale di Fisica Nucleare (INFN) and the Dutch Nikhef, with contributions by institutions from Belgium, Germany, Greece, Hungary, Ireland, Japan, Monaco, Poland, Portugal, Spain.

The authors thank the University of Florida and Columbia University in the City of New York for their generous support.The authors are grateful for the generous support of the National Science Foundation under grant CCF-1740391. The authors thank Sharon Sputz of Columbia University for her effort in facilitating this collaboration. 

I.B. acknowledges the support of the National Science Foundation under grant PHY-1911796 and the Alfred P. Sloan Foundation.

\end{acknowledgments}

\bibliography{Refs} 

\appendix

\section{Proofs of Key Technical Claims} \label{app:proofs}

\subsection{Proof of Proposition \ref{prop:convex-boundary}}

Combining the definitions of the likelihood ratio $\lambda(\mb x)$ and the probability densities $\rho_0(\mb x)$ and $\rho_1(\mb x)$, we have
\begin{align}
    \lambda(\mb x) &= \frac{\int \rho_{\mr{noise}}(\mb x - \mb s_{\mb \gamma} ) \, \mr{d}\nu(\mb \gamma)}{\rho_{\mr{noise}}(\mb x)} \\
    &= \int \frac{\rho_{\mr{noise}} (\mb x - \mb s_{\mb \gamma} )}{{\rho_{\mr{noise}}(\mb x)}} \, \mr{d}\nu(\mb \gamma).
\end{align}
When the noise is Gaussian $\mathcal{N}(\mb 0, \sigma^2 \mb I)$, the integrand equals
\begin{align}
    \frac{\rho_{\mr{noise}} (\mb x - \mb s_{\mb \gamma} )}{{\rho_{\mr{noise}}(\mb x)}} &= \exp \left(\frac{\left<\mb x,\mb s_{\mb\gamma}\right> - \|\mb s_{\mb\gamma}\|^2/2}{\sigma^2} \right),
\end{align}
which is a convex function of $\mb x$. Hence after integrating over $\mb\gamma$, the resulting function $\lambda(\mb x)$ is still a convex function of $\mb x$. The optimal decision region is a sublevel set of $\lambda(\mb x)$, and is hence a convex set.

\subsection{Proof of Proposition \ref{prop:risk-minimization}}

Assume the training data is drawn iid from some distribution on $(\mb x,y)\in\R^n\times\{0,1\}$. In this setting, the previous defined densities $p_0(\mb x)$ and $p_1(\mb x)$ can be expressed as $p_0(\mb x) = p(\mb x|y=0)$ and $p_1(\mb x) = p(\mb x|y=1)$. If the predictor function is $f:\R^n\to\R$, then the risk is
\begin{align}
    \mathcal{R}(f) &= \E_{(\mb x,y)}[\ell(f(\mb x),y)] \\
    &= \prob[y=0] \cdot \E_{\mb x|y=0}[\ell(f(\mb x),0)] \ + \nonumber\\
    &\qquad \prob[y=1] \cdot \E_{\mb x|y=1}[\ell(f(\mb x),1)] \\
    &= \prob[y=0]  \int_{\R^n} \ell(f(\mb x),0) p_0(\mb x) \mathrm{d}\mb x \ +\nonumber\\
    &\qquad \prob[y=1] \int_{\R^n} \ell(f(\mb x),1) p_1(\mb x) \mathrm{d}\mb x \\
    &= \int_{\R^n} \Big( (1-c) \ell(f(\mb x),0) p_0(\mb x) + \nonumber\\
    &\qquad c \ell(f(\mb x),1) p_1(\mb x) \Big) \mathrm{d}\mb x,
\end{align}
where $c:=\prob[y=1] \in (0,1)$ is an exogenous constant that only depends on the data distribution.
The function that minimizes the above risk is
\begin{equation}
    f_\star(\mb x) = \arg\min_{\hat{y}} \  (1-c) \ell(\hat{y},0) p_0(\mb x) + c \ell(\hat{y},1) p_1(\mb x)
\end{equation}
for all $\mb x\in\R^n$, or equivalently
\begin{equation}
    f_\star(\mb x) = \arg\min_{\hat{y}} \  \ell(\hat{y},0)  + \frac{c \lambda(\mb x)}{1-c} \ell(\hat{y},1) .
    \label{eqn:riskmin-lambda}
\end{equation}
Therefore, the optimal predicted value at a point is the solution to an optimization problem that only depends on the likelihood ratio $\lambda(\mb x)$. 

Take an arbitrary fixed $\mb x$. From the assumption that $\ell(\hat{y},y)$ is strictly convex and minimized at $\hat{y}=y$, it follows that $\ell(\hat{y},0) + \frac{c \lambda(\mb x)}{1-c} \ell(\hat{y},1)$ is strictly convex in $\hat{y}$, strictly decreasing on $(-\infty,0]$ and strictly increasing on $[1,\infty)$. Hence for any $\mb x$ the risk minimization problem of equation \eqref{eqn:riskmin-lambda} has a unique solution in $[0,1]$. The optimal solution can be found from the first-order-condition (FOC). Noticing that $\hat{y}$ cannot be 0 or 1 under the FOC, we can rewrite the FOC as
\begin{equation}
    \frac{\ell'(\hat{y},0)}{-\ell'(\hat{y},1)} = \frac{c \lambda(\mb x)}{1-c}. \label{eqn:foc}
\end{equation}
From the assumption of strong convexity, we know that on the interval $(0,1)$ we have $\ell'(\hat{y},0)>0$ and $\ell'(\hat{y},1)<0$, where in $\ell'$ the derivative is taken with respect to the first argument. Hence the left-hand-side of \eqref{eqn:foc} is strictly increasing in $\hat{y}$.

This concludes that the optimal decision function $f_\star(\mb x)$ is strictly increasing in $\lambda(x)$.

\end{document}